\documentclass[english,13pt]{article}
\usepackage[T1]{fontenc}
\usepackage[latin1]{inputenc}
\usepackage{geometry}
\usepackage{float}
\usepackage{mathrsfs}
\usepackage{amsmath}
\usepackage{amssymb}
\usepackage{graphicx}
\usepackage{subfigure}

\makeatletter

\usepackage{algorithm,algpseudocode}

\usepackage{amsthm}

\usepackage{mathrsfs}

\usepackage{amsfonts}

\usepackage{epsfig}

\usepackage{bm}

\usepackage{mathrsfs}

\usepackage{enumerate}

\numberwithin{equation}{section}

\@ifundefined{definecolor}{\@ifundefined{definecolor}
 {\@ifundefined{definecolor}
 {\usepackage{color}}{}
}{}
}{}

\usepackage{subfig}\usepackage[all]{xy}

\newtheorem{theorem}{Theorem}[section]
\newtheorem{lem}{Lemma}[section]

\newcounter{hypA}
\newenvironment{hypA}{\refstepcounter{hypA}\begin{itemize}
  \item[({\bf A\arabic{hypA}})]}{\end{itemize}}
\newcounter{hypB}

\newcounter{hypD}
\newenvironment{hypD}{\refstepcounter{hypD}\begin{itemize}
 \item[({\bf D\arabic{hypD}})]}{\end{itemize}}

\usepackage{babel}\date{}

\usepackage{babel}

\makeatother

\usepackage{babel}

\begin{document}

\begin{center}

{\Large \textbf{Multilevel Particle Filters for a Class of Partially Observed Piecewise Deterministic Markov Processes}}

\vspace{0.5cm}

BY  AJAY JASRA$^{1}$, KENGO KAMATANI$^{2}$ \& MOHAMED MAAMA$^{1}$ 

{\footnotesize $^{1}$Applied Mathematics and Computational Science Program,  Computer, Electrical and Mathematical Sciences and Engineering Division, King Abdullah University of Science and Technology, Thuwal, 23955-6900, KSA.}\\
{\footnotesize $^{2}$Institute of Statistical Mathematics,  Tokyo 190-0014,  JP.}\\
{\footnotesize E-Mail:\,} \texttt{\emph{\footnotesize ajay.jasra@kaust.edu.sa, kamatani@ism.ac.jp, maama.mohamed@gmail.com}}

\end{center}

\begin{abstract}
In this paper we consider the filtering of a class of partially observed piecewise deterministic Markov processes (PDMPs).
In particular, we assume that an ordinary differential equation (ODE) drives the deterministic element and can only be solved numerically via a time discretization. We develop, based upon the approach in \cite{lemaire}, a new particle and multilevel particle filter (MLPF) in order to approximate the filter associated to the discretized ODE. We provide a bound on the mean square error associated to the MLPF which provides guidance on setting the simulation parameter of that algorithm and implies that significant computational gains can be obtained versus using a particle filter. Our theoretical claims are confirmed in several numerical examples.
\end{abstract}

\section{Introduction}

Piecewise deterministic Markov processes (PDMPs) are a class of continuous-time stochastic processes that jump at random times and in-between jump times evolve deterministically. They have been considered in a variety of applications such as target tracking and neuroscience; see e.g.~\cite{bunch,godsill,lemaire}. In this article we consider the case where the deterministic component evolves according to an ordinary differential equation (ODE) which can only be solved numerically, using time discretization (e.g.~the Euler method). In addition, as in \cite{finke}, we consider that the process is only partially observed at discrete and regular times.

The filtering, i.e.~the recursive estimation of the hidden process at pre-specified times,
of discretely observed continuous-time processes is a notoriously challenging problem; see e.g.~\cite{bain,delm_book} for several examples. The main barrier to their use is the lack of analytical tractability of the filter and this gives rise to the application of numerical methods. In particular, one of the most often used methods in such context is the particle filter (e.g.~\cite{delm_book}), which is a numerical technique that generates a collection of $S\in\mathbb{N}$ samples in parallel, and which undergo resampling and sampling operations. Such methods can be used to approximate expectations w.r.t.~the filtering distribution and these approximations are asymptotically consistent (i.e.~as $S\rightarrow+\infty$ the approximations converge in an appropriate sense); see \cite{delm_book}.

In the case of piecewise deterministic processes (PDPs) several particle filter and sequential Monte Carlo (e.g.~\cite{delm:06}) methods have appeared, for example in \cite{bunch,godsill,whiteley}. In this article, as mentioned previously, we focus upon a particular type of PDMP, which was extensively investigated in \cite{lemaire}. This PDMP includes an ODE which can only be solved by using numerical methods and we will develop particle filters for this type of model. The particle filter we use is essentially the most basic bootstrap approach, but with the caveat that one is using a time-discretized solution of the ODE. This latter point naturally points to using the well-known multilevel Monte Carlo (MLMC) method \cite{giles,giles1,hein} that can help to improve estimation of expectations where the associated probability law is subject to time discretization; this is a point that had been realized by \cite{lemaire}.

MLMC methods work by considering a collection of time-discretized expectations of increasingly more accurate approximations. The idea is then to consider a telescoping sum of difference of the expectations, equal to the most precise approximation and then to use Monte Carlo methods to approximate each difference. The key point of the methods is to perform exact simulation from a coupling of the probability laws at consecutive levels. If such a coupling is good enough (see e.g.~\cite{giles}) the cost to compute the expectation of interest for a given mean square error (MSE) between the estimator to the true expectation, relative to a using just one level, can be reduced by several orders of magnitude. This idea has been extended from forward problems (without data), as are considered in \cite{giles,lemaire}, to models which combine the forward problems with real data; see \cite{beskos,mlpf,anti_pf,mlpf1,levymlpf,ml_rev}. Despite the work in the afore-mentioned articles, we do not know of any method that has dealt with the type of partially observed PDMP considered in this paper.

In this article, we develop a new multilevel particle filter (MLPF) for the filtering of a class of PDMPs. Critical to our approach is the clever change of measure method that is used in \cite{lemaire}. In that article, the authors consider the PDMP by itself and considering estimating differences of expectations at different levels of accuracy in terms of the numerical approximation of the ODE solver. They realize that in order to couple well samples of the PDMP at consecutive levels, the processes should jump at the same times. In order to ensure that this is so, they use a change of measure at the courser time discretization, so that the process jumps at the same time as the finer time discretization. We show how this method can be embedded within the MLPF framework that was originally developed in \cite{mlpf}. We also consider a mathematical analysis of our new filtering estimator.
Our mathematical results show how to choose the maximum level and the number of samples needed at each consecutive level difference, so as to obtain a MSE of $\mathcal{O}(\epsilon^2)$, $\epsilon>0$. The MSE is the difference of our estimator and the true (without time discretization) filtering expectation. If $\Delta_l=2^{-l}$ represents the time discretization of the ODE solver, typically the cost of using many solvers (e.g.~the Euler method) is $\mathcal{O}(\Delta_l^{-1})$. If one can assume this is the cost of sampling the time-discretized PDMP - and this is uncertain - then our theoretical results indicate that the cost to achieve an MSE of $\mathcal{O}(\epsilon^2)$ is $\mathcal{O}(\epsilon^{-2}\log(\epsilon)^2)$ versus using a particle filter, which would cost $\mathcal{O}(\epsilon^{-3})$ for the same MSE. We then investigate these claims numerically in several examples.

This article is structured as follows. In Section \ref{sec:method} we give precise details on the model and our algorithm.
In Section \ref{sec:theory} we present our mathematical results. In Section \ref{sec:numerics} we investigate our approach 
in several numerical examples. The appendix houses all of our mathematical proofs with assumptions.

\section{Model and Methodology}\label{sec:method}

\subsection{Process}

The presentation of the current section follows the exposition of \cite{lemaire} very closely,  as this latter work focusses exactly on the processes that are considered in that article.

Roughly, a PDMP is continuous-time stochastic process that will move deterministically in-between random (jump) times. At the jump times there may also be transitions of the stochastic process.  More precisely, we will need the following elements:
\begin{itemize}
\item{$\mathsf{X}:=\mathsf{U}~\times~\mathbb{R}^d$ as the space of the PDMP, where $\mathsf{U}$ is a countable set, with possible infinite cardinality and $d\in\mathbb{N}$ is fixed. Also, let $\mathscr{X}$ be its Borel $\sigma$-algebra.}
\item{A differential equation, for any fixed $x\in\mathsf{X}$:
\begin{equation}\label{eq:ode}
\nabla_t\Phi(x,t) = f\left(\Phi(x,t)\right)
\end{equation}
where $\Phi:\mathsf{X}~\times~\mathbb{R}^+\rightarrow\mathsf{X}$,  $f:\mathsf{X}\rightarrow\mathsf{X}$,
$\Phi(x,0)=x$ is given.  
It is assumed that, for any fixed $x\in\mathsf{X}$, there exists a unique solution to \eqref{eq:ode} and conditions to ensure this are discussed later on.
}
\item{A bounded and measurable function $\lambda:\mathsf{X}\rightarrow(0,\lambda^{\star}]$, $0<\lambda^{\star}<+\infty$.}
\item{A transition kernel $Q:\mathsf{X}\times\mathscr{X}\rightarrow[0,1]$. 
}
\end{itemize}
The upper-bound $\lambda^{\star}$ is a necessary component of the methodology to be described.

The $\mathsf{X}$-valued PDMP, $(X_t)_{t\in[0,T]}$, is then generated on a time interval $[0,T]$, $X_t=(U_t,V_t)$, with $X_0=(u_0,v_0)$ given as described in Algorithm \ref{alg:pdmp}, with $0=t_0$ and $T=t_1$. In Algorithm \ref{alg:pdmp} we use the notation $\mathcal{U}_{[0,1]}$ to denote the Uniform distribution on $[0,1]$ and $\mathcal{E}(\lambda^{\star})$ to denote the exponential distribution of parameter $\lambda^{\star}$.
Throughout, we assume $T\in\mathbb{N}$ is `total length' of the PDMP and that $T$ is so large that performing filtering (to be defined later on) that essentially one may not reach $T$ in a practical problem. This is not necessarily a challenge as time can be rescaled to reflect this assumption.

\begin{algorithm}
\begin{enumerate}
\item{Set $k=0$, $k^{\star}=1$,  initial time $0\leq t_0$, final time $t_0< t_1\leq T$, $T^{\star}_0=T_0=t_0
$, and initial point $x_{t_0}=(u_{t_0},v_{t_0})$.
}
\item{Generate $I_{k^{\star}}^{\star}\sim\mathcal{E}(\lambda^{\star})$ and set $T^{\star}_{k^{\star}}=T^{\star}_{k^{\star}-1}+I_{k^{\star}}^{\star}$. If $T^{\star}_{k^{\star}}>t_1$ set $X_s=\Phi(X_{T_{k}},s-T_{k})$ for every $s\in(T_{k},t_1]$
and go to step 4., otherwise go to the next step.}
\item{Generate $W_{k^{\star}}\sim\mathcal{U}_{[0,1]}$ 
\begin{itemize}
\item{
If
\begin{equation}\label{eq:alg_dec}
W_{k^{\star}} \leq \frac{\lambda\left(\Phi(x_{T_k},T^{\star}_{k^{\star}}-T_k)\right)}{\lambda^{\star}}
\end{equation}
then set $k\leftarrow k+1$, $\tau_k=k^{\star}$, 
$T_k=T^{\star}_{k^{\star}}$, $X_s=\Phi(X_{T_{k-1}},s-T_{k-1})$ for every $s\in(T_{k-1},T_k)$. Generate $X_{T_k}\sim Q\left(X_{T_k-},~\cdot~\right)$. Set $k^{\star}\leftarrow k^{\star}+ 1$ and go to step 2..}
\item{Otherwise set $k^{\star}\leftarrow k^{\star}+ 1$ and go to step 2..}
\end{itemize}
}
\item{Return the process $(X_t)_{t\in[t_0,t_1]}$ and $N_{[t_0,t_1]}=k$ and $N^{\star}_{[t_0,t_1]}=k^{\star}-1$.}
\end{enumerate}
\caption{Simulation of a PDMP.  Some notations are defined in the main text.}
\label{alg:pdmp}
\end{algorithm}

The expression of the process in Algorithm \ref{alg:pdmp} is based upon a Poisson thinning technique explained in detailed in \cite{lemaire}. Therefore, the upper-bound of $\lambda$ is critical since the jump times $T_k$ are generated by thinning a Poisson process jumps $T^*_{k}$ of rate $\lambda^*$.   In practice, of course, it is sufficient to work with the process at jump times and `fill in' any points that are necessary.  In Algorithm \ref{alg:pdmp}, $N_{[t_0,t_1]}$ and $N_{[t_0,t_1]}^{\star}$ represent the number of jump times on $[t_0,t_1]$ of the PDMP and of a Poisson process of rate $\lambda^{\star}$. If $t_0=0$, we write $N_{t_1}$ and $N_{t_1}^{\star}$ as a short-hand.
The simulation in Algorithm \ref{alg:pdmp} also requires that one can solve the equation \eqref{eq:ode}, which we shall assume is not possible exactly. 

\subsubsection{Uniqueness and Discretization of the Solution to \eqref{eq:ode}}

In order to proceed further, we shall assume the following.
\begin{hypD}\label{hyp:d1}
$f:\mathsf{X}\rightarrow\mathsf{X}$ as on the R.H.S.~of \eqref{eq:ode} has the following properties:
\begin{itemize}
\item{$\sup_{x\in\mathsf{X}}|f(x)|<+\infty$, where $|\cdot|$ is the $\mathbb{L}_1-$norm.}
\item{$f(u,\cdot)$ is globally differentiable with continuous first derivative for each $u\in\mathsf{U}$.}
\item{There exists a $C<+\infty$ such that for any $u\in\mathsf{U}, (v,v')\in\mathbb{R}^{2d}$, $|f(u, v)-f(u, v')|\leq C|v-v'|$.}
\item The flow does not change the value of $u$, that is, $\pi_{\mathsf{U}}(f(x))=0$ for any $x\in\mathsf{X}$, where $\pi_{\mathsf{U}}(u,v)=u$. 
\end{itemize}  
The Markov kernel $Q$ satisfies the following: 
\begin{itemize}
    \item The jump does not change the value of $v$, that is, $Q(x,A^c)=0$ where $A=\{(u^*,v^*): v^*\neq v\}$ and $x=(u,v)$. 
\end{itemize}
\end{hypD}
Assumption (D\ref{hyp:d1}) is enough to ensure a unique solution to \eqref{eq:ode}. It is assumed to hold from herein and, for instance, is omitted from any statements of our mathematical results.

We now proceed to introduce a time discretization to solve \eqref{eq:ode}.  We adopt the well-known Euler method,
although several more advanced approaches such as (higher-order) Runge--Kutta methods could be employed with little difficulty.  Let $\Delta_l=2^{-l}$, $l\in\mathbb{N}$ be given and consider two time points $t_0,t_1$ such that $0\le t_0<t_1\le T$. Let 
$J:=\max\{\lfloor(t_1-t_0)\Delta_l^{-1}\rfloor-1,0\}$. For each $x\in\mathsf{X}$ and an integer $0\le j\le J$, consider a discretized flow of $\Phi$ by 
$$
\phi^l(x,(j+1)\Delta_l) = \phi^l(x,j\Delta_l) + f(\phi^l(x,j\Delta_l))\Delta_l.
$$
Then our approximate solution is taken as, for $x\in\mathsf{X}, 0\le j\le J$ and for $j\Delta_l\le t< (j+1)\Delta_l$
\begin{equation}
\label{eq:euler}
\Phi_{[t_0,t_1]}^l(x,t) = \phi^l(x,j\Delta_l) + f(\phi^l(x,t_0+j\Delta_l))~\{t-j\Delta_l\}.
\end{equation}
Then to generate a discretized approximation of the process in Algorithm \ref{alg:pdmp}, one can simply replace the solution of \eqref{eq:ode}, with the discretized solution given in \eqref{eq:euler}. More precisely, in \eqref{eq:alg_dec}
one uses $\Phi_{[T_k^l,T^{\star,l}_{k^{\star,l}}]}^l(x_{T_k^l}^l,T^{\star,l}_{k^{\star,l}}-T_k^l)$ instead of 
$\Phi(x_{T_k},T^{\star}_{k^{\star}}-T_k)$. To further clarify,
we denote the piecewise deterministic process (PDP) that is produced by Algorithm \ref{alg:pdmp}, when replacing the solution of \eqref{eq:ode}
with the solution of \eqref{eq:euler} as $(X_t^l)_{t\in[t_0,t_1]}$ and the all the associated notations will also be given a supersrcipt $l$. For instance the event times of the homogeneous Poisson process will be denoted $(T_1^{\star,l},\dots,T_{N_{[t_0,t_1]}^{\star,l}}^{\star,l})$.  We remark that, as explained in \cite[Section 2.3]{lemaire}, whilst the process in Algorithm \ref{alg:pdmp} indeed defines a PDMP, the associated discretized approximation, is only a PDP.

\subsubsection{Coupling Discretizations}

As it will be critical when describing Multilevel methods, consider $\varphi:\mathsf{X}\rightarrow\mathbb{R}$,  with $\varphi$ bounded and measurable (we write such functions as $\mathcal{B}_b(\mathsf{X})$ from herein),
and simulating two discretized processes $(X_t^l)_{t\in[t_0,t_1]}$ and $(X_t^{l-1})_{t\in[t_0,t_1]}$,  $l\in\{2,3,\dots\}$, $0\leq t_0<t_1\leq T$,  where we shall emphasize that the initial points need not be identical. Suppose that one wishes to approximate:
$$
\mathbb{E}^l_{x_{t_0}^l}[\varphi(X_{t_1}^l)] - \mathbb{E}^{l-1}_{x_{t_0}^{l-1}}[\varphi(X_{t_1}^{l-1})]
$$
where $\mathbb{E}_{x}^s[~\cdot~]$, $s\in\{l-1,l\}$, denotes an expectation w.r.t.~the law of the process $(X_t^s)_{t\in[t_0,t_1]}$ starting from $X_{t_0}^s=x$. Clearly, one approach is to sample the processes $(X_t^s)_{t\in[t_0,t_1]}$,  $s\in\{l-1,l\}$, $N-$times ($N\in\mathbb{N}$) independently and then use Monte Carlo approximation. However, in order for MLMC methods to work well, it is important to sample a dependent coupling of the joint laws of $(X_t^s)_{t\in[t_0,t_1]}$,  $s\in\{l-1,l\}$, or to sample one process and correct using importance sampling; this is the latter strategy that is employed in \cite{lemaire} and the one we shall follow.

Below, the following notations will be used and will help to facilitate the description of our algorithms. We set, for 
$l\in\{2,3,\dots\}$
\begin{eqnarray*}
\Xi_{t_0,t_1}^l & := & \left(N_{[t_0,t_1]}^{\star,l}, N_{[t_0,t_1]}^l, (T_0^{\star,l},\dots,T_{N_{[t_0,t_1]}^{\star,l}}^{\star,l}),(\tau_1^l,\dots,\tau_{N_{[t_0,t_1]}^l}^l), (X^l_{T_1^l},\dots,X^l_{T_{N^l_{[t_0,t_1]}}}),
(X^{l-1}_{T_1^l},\dots,X^{l-1}_{T_{N^l_{[t_0,t_1]}}})\right) \\
\mathbf{X}_{t_0}^l & := & (X_{t_0}^l,X_{t_0}^{l-1}) \\
\mathbf{X}_{t_1}^l & := & (X_{t_1}^l,X_{t_1}^{l-1}).
\end{eqnarray*}
We note that in $\Xi_{t_0,t_1}^l$ the jump times $(T_1^l,\dots,T_{N_{[t_0,t_1]}^l}^{l})$ of the PDP can be inferred on the basis of $(T_0^{\star,l},\dots,T_{N_{[t_0,t_1]}^{\star,l}}^{\star,l}),(\tau_1^l,\dots,\tau_{N_{[t_0,t_1]}^l}^l)$ and hence are not added explicitly (except in some subscripts, for readability).

One of the approaches of \cite{lemaire}, which is the one we shall focus upon, is simply to generate the process 
$(X_t^l)_{t\in[t_0,t_1]}$ and then correct the expectation $\mathbb{E}_{x_{t_0}^{l-1}}[\varphi(X_{t_1}^{l-1})]$ appropriately. More precisely, one has the identity:
\begin{equation}\label{eq:id}
\mathbb{E}^l_{x_{t_0}^l}[\varphi(X_{t_1}^l)] - \mathbb{E}^{l-1}_{x_{t_0}^{l-1}}[\varphi(X_{t_1}^{l-1})] = 
\mathbb{E}^l_{x_{t_0}^l}[\varphi(X_{t_1}^l)] - \mathbb{E}^l_{x_{t_0}^{l}}\left[\varphi(X_{t_1}^{l-1})R_{t_0,t_1}^l(\mathbf{X}_{t_0}^l,\Xi_{t_0,t_1}^l,\mathbf{X}_{t_1}^l)\right]
\end{equation}
where $R_{t_0,t_1}^l(\mathbf{X}_{t_0}^l,\Xi_{t_0,t_1}^l,\mathbf{X}_{t_1}^l)$ is defined via (see \cite[Corollary 2.2]{lemaire}):
\begin{eqnarray}
R_{t_0,t_1}^l(\mathbf{X}_{t_0}^l,\Xi_{t_0,t_1}^l,\mathbf{X}_{t_1}^l) & = & Z_{N_{[t_0,t_1]}^l}^l\prod_{p=0}^{N_{[t_0,t_1]}^l-1}Z_p^l\label{eq:rt_def1}
\end{eqnarray}
where
\begin{eqnarray}
Z_p^l  = 
\frac{\tfrac{\lambda\left(\Phi_{[T_{p}^{l},T_{p+1}^{l}]}^{l-1}(x_{T_{p}^{l}}^{l-1},T_{p+1}^{l}-T_{p}^{l})\right)}{\lambda^{\star}}\prod_{q=\tau_p^l+1}^{\tau_{p+1}^l-1}\left\{
1-
\tfrac{\lambda\left(\Phi_{[T_{p}^{l},T_{q}^{\star,l}]}^{l-1}(x_{T_{p}^{l}}^{l-1},T_{q}^{\star,l}-T_{p}^{l})\right)}{\lambda^{\star}}
\right\}}{
\tfrac{\lambda\left(\Phi^{l}_{[T_{p}^{l},T_{p+1}^{l}]}(x_{T_{p}^{l}}^{l},T_{p+1}^{l}-T_{p}^{l})\right)}{\lambda^{\star}}\prod_{q=\tau_p^l+1}^{\tau_{p+1}^l-1}\left\{
1-\tfrac{
\lambda\left(\Phi_{[T_{p}^{l},T_{q}^{\star,l}]}^{l}(x_{T_{p}^{l}}^{l},T_{q}^{\star,l}-T_{p}^{l})\right)
}{\lambda^{\star}}
\right\}
} 
\frac{Q\left(x_{T_{p+1}^{l}-}^{l-1},x_{T_{p+1}^{l}}^{l-1}\right)}{Q\left(x_{T_{p+1}^{l}-}^l,x_{T_{p+1}^{l}}^l\right)}
\end{eqnarray}
for $p\in\{0,\dots,N_{[t_0,t_1]}^l-1\}$ and 
\begin{eqnarray}
Z_{N_{[t_0,t_1]}^l}^l & = &
\frac{\prod_{q=\tau_{N_{[t_0,t_1]}^{l}}^l+1}^{N_{[t_0,t_1]}^{\star,l}}\left\{
1-
\tfrac{\lambda\left(\Phi_{[T_{N_{[t_0,t_1]}^{l}}^{l},T_{q}^{\star,l}]}^{l-1}(x_{T_{N_{[t_0,t_1]}^{l}}^{l}}^{l-1},T_{q}^{\star,l}-
T_{N_{[t_0,t_1]}^l}^{l})\right)}{\lambda^{\star}}
\right\}}{
\prod_{q=\tau_{N_{[t_0,t_1]}^l}^l+1}^{N_{[t_0,t_1]}^{\star,l}}\left\{
1-\tfrac{
\lambda\left(\Phi_{[T_{N_{[t_0,t_1]}^{l}}^{l},T_{q}^{\star,l}]}^{l}(
x_{T_{N_{[t_0,t_1]}^{l}}^l}^{l},T_{q}^{\star,l}-T_{N_{[t_0,t_1]}^l}^{l})\right)
}{\lambda^{\star}}
\right\}
}
\label{eq:rt_def4}
\end{eqnarray}
where we have used the abuse of notation that $Q$ denotes both the kernel and the positive density (w.r.t.~ the product of the Lebesgue measure and the counting measure) of $Q$, and that for $p\in\{1,\dots,N_{[t_0,t_1]}^l\}$, $u_{T_{p}^l}^{l-1}=u_{T_{p}^l}^{l}$.

In words \eqref{eq:id} essentially says that the jump times of the two processes are same, as are the values of the discrete process (except possibly at the initial time) and that $R_{t_0,t_1}^l$ is the Radon--Nikodym derivative used to account for the discrepancy of the target process at level $l-1$ and the sampled process. We are using a slight abuse of notation as the $(X_t^{l-1})_{t\in(t_0,t_1]}$ that is used in the expectation operator on the R.H.S.~of \eqref{eq:id} is not the one as if one had run Algorithm \ref{alg:pdmp} with $\Phi^{l-1}$, but the one as described previously. The main reason for this is to avoid notational overload.

\subsection{Filter and Multilevel Identity}

We begin by considering the filter associated to a partially observed PDMP, followed by the associated time discretized filter and finally the multilevel identity that we intend to approximate.

We will consider discrete time data that are regularly observed at unit times; this latter hypothesis is for notational convenience and, indeed, one can modify the ideas of this article to any regular time interval, between data.
Our data are observations of the sequence of random variables $(Y_1,Y_2,\dots)$, with $Y_n\in\mathsf{Y}$, with the latter space being left abstract for now.  For $p\in\mathbb{N}$, we will write $M_p:\mathsf{X}\times\mathscr{X}\rightarrow[0,1]$, where $\mathscr{X}$ is the $\sigma-$field generated by $\mathsf{X}$, as the transition kernel induced by the exact PDMP (i.e.~as described in Algorithm \ref{alg:pdmp}) from time $p-1$ to time $p$. Then we shall assume for any $(n,A)\in\mathbb{N}\times\mathscr{Y}$ ($\mathscr{Y}$ is the $\sigma-$ field generated by $\mathsf{Y}$)
that 
$$
\mathbb{P}(Y_n\in A\mid\{X_t\}_{t\in[0,n]},y_1,\dots,y_{n-1},y_{n+1},\dots) = \int_{A} g(x_n,y)dy
$$
where $dy$ is a dominating measure (often Lebesgue).
Then,  we can define the filter as, for $(n,\varphi)\in\mathbb{N}\times\mathcal{B}_b(\mathsf{X})$
$$
\pi_n(\varphi) := \frac{\int_{\mathsf{X}^n}\varphi(x_n)\left\{\prod_{p=1}^n g(x_p,y_p)\right\}\prod_{p=1}^n M_p(x_{p-1},dx_p)}{
\int_{\mathsf{X}^{n}}\left\{\prod_{p=1}^n g(x_p,y_p)\right\}\prod_{p=1}^n M_p(x_{p-1},dx_p)}
$$

As we cannot work with the PDMP directly,  we denote for $(p,l)\in\mathbb{N}^2$, $M_p^l:\mathsf{X}\times\mathscr{X}\rightarrow[0,1]$ as the transition kernel of the time-discretized approximation from time $p-1$ to time $p$ of the PDMP. Then,
we will define the following approximation, which is now our focus,
for $(n,\varphi,l)\in\mathbb{N}\times\mathcal{B}_b(\mathsf{X})\times\mathbb{N}$
$$
\pi_n^l(\varphi) := \frac{\int_{\mathsf{X}^n}\varphi(x_n)\left\{\prod_{p=1}^n g(x_p,y_p)\right\}\prod_{p=1}^n M_p^l(x_{p-1},dx_p)}{
\int_{\mathsf{X}^{n}}\left\{\prod_{p=1}^n g(x_p,y_p)\right\}\prod_{p=1}^n M_p^l(x_{p-1},dx_p)}.
$$
We shall show, in Lemma \ref{lem:7} in the appendix, that for $(n,\varphi)\in\mathbb{N}\times\mathcal{B}_b(\mathsf{X})$
$$
\lim_{l\rightarrow\infty} \pi_n^l(\varphi) = \pi_n(\varphi).
$$
Now for $L\in\mathbb{N}$, given,  our objective is to approximate the identity, for $(n,\varphi)\in\mathbb{N}\times\mathcal{B}_b(\mathsf{X})$
\begin{equation}\label{eq:ml_id}
\pi_n^L(\varphi) = \pi_n^1(\varphi) + \sum_{l=2}^L\left\{\pi_n^l(\varphi)-\pi_n^{l-1}(\varphi)\right\}.
\end{equation}
Now,  we can take advantage of the identity \eqref{eq:id} to assist in the approximation of \eqref{eq:ml_id}.
First, consider the Radon--Nikodym derivative $R_{t_0,t_1}^l$ as defined in the system of equations \eqref{eq:rt_def1}-\eqref{eq:rt_def4}; we shall consider this object over a unit time.  Second, write the expectation w.r.t.~the discretized process at level $l$ as $\mathbb{E}^l$. Then it is simple to show that
\begin{equation}\label{eq:ml_id_com}
\pi_n^l(\varphi)-\pi_n^{l-1}(\varphi) = \frac{\mathbb{E}^l\left[\varphi(X_n^l)\left\{\prod_{p=1}^n g(X_p^l,y_p)\right\}\right]}{\mathbb{E}^l\left[\left\{\prod_{p=1}^n g(X_p^l,y_p)\right\}\right]} -
\frac{\mathbb{E}^l\left[\varphi(X_n^{l-1})\left\{\prod_{p=1}^n R_{p-1,p}^l(\mathbf{X}_{p-1}^l,\Xi_{p-1,p}^l,\mathbf{X}_{p}^l)g(X_p^l,y_p)\right\}\right]}{\mathbb{E}^l\left[\left\{\prod_{p=1}^n R_{p-1,p}^l(\mathbf{X}_{p-1}^l,\Xi_{p-1,p}^l,\mathbf{X}_{p}^l)g(X_p^l,y_p)\right\}\right]}.
\end{equation}
Note that in the expectation operator including the terms $\mathbf{X}_{p-1}^l$ and $\mathbf{X}_{p}^l$, that across the levels, the random variables are equal (e.g.~$X_{p-1}^l=X_{p-1}^{l-1}$); this may not be the case in the approximation that we will generate and is thus the reason for this notation.
Therefore,  we can rewrite \eqref{eq:ml_id} as 
\begin{eqnarray}
\pi_n^L(\varphi) & = & \pi_n^1(\varphi) + \sum_{l=2}^L\Bigg\{
\frac{\mathbb{E}^l\left[\varphi(X_n^l)\left\{\prod_{p=1}^n g(X_p^l,y_p)\right\}\right]}{\mathbb{E}^l\left[\left\{\prod_{p=1}^n g(X_p^l,y_p)\right\}\right]} - \nonumber\\ & &
\frac{\mathbb{E}^l\left[\varphi(X_n^{l-1})\left\{\prod_{p=1}^n R_{p-1,p}^l(\mathbf{X}_{p-1}^l,\Xi_{p-1,p}^l,\mathbf{X}_{p}^l)g(X_p^l,y_p)\right\}\right]}{\mathbb{E}^l\left[\left\{\prod_{p=1}^n R_{p-1,p}^l(\mathbf{X}_{p-1}^l,\Xi_{p-1,p}^l,\mathbf{X}_{p}^l)g(X_p^l,y_p)\right\}\right]}\Bigg\}.\label{eq:ml_id_new}
\end{eqnarray}
We shall explain how to approximate the R.H.S.~of \eqref{eq:ml_id_new} and why it is preferable to considering
\eqref{eq:ml_id} in the next section. 

\subsection{Particle Filter and Multilevel Particle Filter}

We begin by considering the particle filter, which can be used to approximate $\pi_n^1(\varphi)$ for any $n\in\mathbb{N}$. It is described in Algorithm \ref{alg:pf}. To approximate $\pi_n^1(\varphi)$,  at step 2.~of
Algorithm \ref{alg:pf}, once $(G_n^{1,1},\dots,G_n^{S,1})$ has been computed, but before resampling has been performed, one can use the estimate:
$$
\pi_n^{S,1}(\varphi) := \sum_{i=1}^{S} G_n^{i,1}\varphi(X_n^{i,1}).
$$

\begin{algorithm}
\begin{enumerate}
\item{Initialize: Generate $(X_1^{1,1},\dots,X_1^{S,1})$ independently from $M_1^1(x_0,\cdot)$. Set $n=1$ and go to step 2..}
\item{Resampling: For $i\in\{1,\dots,S\}$ compute
$$
G_n^{i,1} = \frac{g(x_n^{i,1},y_n)}{\sum_{j=1}^{S} g(x_n^{j,1},y_n)}.
$$
Sample with replacement from $(x_n^{1,1},\dots,x_n^{S,1})$, using the $(G_n^{1,1},\dots,G_n^{S,1})$ and denote the resulting samples $(\hat{x}_n^{1,1},\dots,\hat{x}_n^{S,1})$ also.  Go to step 3..}
\item{Sampling: For $i\in\{1,\dots,S\}$ sample, conditionally independently,  $X_{n+1}^{i,1}\mid\hat{x}_n^{i,l}$ from 
$M_{n+1}^1(x_n^{i,1},\cdot)$. 
Set $n=n+1$ and go to step 2..}
\end{enumerate}
\caption{Particle Filter}
\label{alg:pf}
\end{algorithm}

To approximate the R.H.S.~of \eqref{eq:ml_id_com}, one can use Algorithm \ref{alg:mlpf}, which is simply a version of the algorithm developed in \cite{mlpf}. To approximate $\pi_n^l(\varphi)-\pi_n^{l-1}(\varphi)$,  at step 2.~of
Algorithm \ref{alg:mlpf}, once $(G_n^{1,l},\dots,G_n^{S_l,l})$ and $(\overline{G}_n^{1,l-1},\dots,\overline{G}_n^{S_l,l-1})$ have been computed, but before resampling has been performed, one can use the estimate:
$$
\pi_n^{S_l,l}(\varphi) - \overline{\pi}_n^{S_l,l-1}(\varphi) := \sum_{i=1}^{S_l} G_n^{i,l}\varphi(X_n^{i,l}) - \sum_{i=1}^{S_l} \overline{G}_n^{i,l-1}\varphi(\overline{X}_n^{i,l-1}).
$$

Therefore our procedure for estimating $\pi_p(\varphi)$ recursively in time is:
\begin{enumerate}
\item{Run Algorithm \ref{alg:pf}.}
\item{For each level $l\in\{2,\dots,L\}$ independently of step 1.~and all other levels, run Algorithm \ref{alg:mlpf}.}
\end{enumerate}
Thus, our multilevel estimator of $\pi_n(\varphi)$ is
$$
\pi_n^{S_{1:L}}(\varphi) := \pi_n^{S_1,1}(\varphi) + \sum_{l=2}^L\left\{\pi_n^{S_l,l}(\varphi) - \overline{\pi}_n^{S_l,l-1}(\varphi)\right\}.
$$
What remains is how to choose $L$ and $(S_1,\dots,S_L)$, which is the topic of the next section. Below we use $\mathbb{\overline{E}}$ to denote expectation w.r.t.~the law of the simulated algorithm as described in the above points 1.~\& 2..

\begin{algorithm}
\begin{enumerate}
\item{Initialize: Generate $((X_1^{1,l},\Xi_{0,1}^{1,l}),\dots,(X_1^{S_l,l},\Xi_{0,1}^{S_l,l}))$ using the time discretized version of Algorithm \ref{alg:pdmp} on $[0,1]$.  For $i\in\{1,\dots,S_l\}$, compute $\overline{x}_1^{i,l-1}$ and
$\mathbf{x}_{0}^{i,l}=(x_0,x_0)$,  $\mathbf{x}_{1}^{i,l}=(x_1^{i,l},\overline{x}_1^{i,l-1})$.
 Set $n=1$ and go to step 2..}
\item{Resampling: For $i\in\{1,\dots,S_l\}$ compute
\begin{eqnarray*}
G_n^{i,l} & = & \frac{g(x_n^{i,l},y_n)}{\sum_{j=1}^N g(x_n^{j,l},y_n)} \\
\overline{G}_n^{i,l-1} & = & \frac{g(\overline{x}_n^{i,l-1},y_n)R_{n-1,n}^l(\mathbf{x}_{n-1}^{i,l},\Xi_{n-1,n}^{i,l},\mathbf{x}_{n-1}^{i,l})}{\sum_{j=1}^N g(\overline{x}_n^{j,l-1},y_n)R_{n-1,n}^l(\mathbf{x}_{n-1}^{j,l},\Xi_{n-1,n}^{j,l},\mathbf{x}_{n-1}^{j,l})}.
\end{eqnarray*}
Sample with replacement from $(x_n^{1,l},\dots,x_n^{S_l,l})$, and 
$(\overline{x}_{n}^{1,l-1},\dots,\overline{x}_{n}^{S_l,l-1})$
 using a maximal coupling of the $(G_n^{1,l},\dots,G_n^{S_l,l})$ and $(\overline{G}_n^{1,l-1},\dots,\overline{G}_n^{S_l,l-1})$
 (see Algorithm \ref{alg:max_coup})
and denote the resulting samples 
$(\hat{x}_n^{1,l},\dots,\hat{x}_n^{S_l,l})$, and 
$(\hat{\overline{x}}_{n}^{1,l-1},\dots,\hat{\overline{x}}_{n}^{S_l,l-1})$ also, 
with $\mathbf{x}_n^{i,l}=(\hat{x}_n^{i,l},\hat{\overline{x}}_n^{i,l-1})$, $i\in\{1,\dots,S_l\}$.  Go to step 3..
}
\item{Sampling: Generate $((X_{n+1}^{1,l},\Xi_{n,n+1}^{1,l}),\dots,(X_{n+1}^{S_l,1},\Xi_{n,n+1}^{S_l,l}))$ using the time discretized version of Algorithm \ref{alg:pdmp} on $[n,n+1]$ with starting points $(\hat{x}_n^{1,l},\dots,\hat{x}_n^{S_l,l})$.  For $i\in\{1,\dots,S_l\}$, compute $\overline{x}_{n+1}^{i,l-1}$ and, 
$\mathbf{x}_{n+1}^{i,l}=(x_{n+1}^{i,l},\overline{x}_{n+1}^{i,l-1})$.
 Set $n=n+1$ and go to step 2..}
\end{enumerate}
\caption{Coupled Particle Filter}
\label{alg:mlpf}
\end{algorithm}

\begin{algorithm}
\begin{enumerate}
\item{Input: $(D_1^{1},\dots,D_1^S), (D_2^1,\dots,D_2^{S})$ and associated probabilities 
$(W_1^{1},\dots,W_1^S), (W_2^1,\dots,W_2^{S})$.}
\item{For $i\in\{1,\dots,S\}$ generate $\tilde{U}\sim\mathcal{U}_{[0,1]}$ (uniform distribution on $[0,1]$)
\begin{itemize}
\item{If $\tilde{u}<\sum_{i=1}^S \min\{W_1^i,W_2^i\}$ generate $a^i\in\{1,\dots,S\}$ using the probability mass function
$$
\mathbb{P}(i) = \frac{\min\{W_1^i,W_2^i\}}{\sum_{j=1}^S \min\{W_1^j,W_2^j\}}
$$
and set $\tilde{D}_j^i=D_j^{a^i}$, 
$j\in\{1,2\}$.}
\item{Otherwise generate $(a_1^i,a_2^i)\in\{1,\dots,S\}^2$ independently via the probability mass functions:
$$
\mathbb{P}_j(i) = \frac{W_j^i-\min\{W_1^i,W_2^i\}}{\sum_{k=1}^S[W_j^k-\min\{W_1^k,W_2^k\}]}
$$
and set $\tilde{D}_j^i=D_j^{a_j^i}$,  $j\in\{1,2\}$.}
\end{itemize}
}
\item{Set: $D_j^i=\tilde{D}_j^{i}$, $(i,j)\in\{1,\dots,S\}\times\{1,2\}$.}
\item{Output: $(D_1^{1},\dots,D_1^S), (D_2^1,\dots,D_2^{S})$.}
\end{enumerate}
\caption{Maximal Coupling Resampling}
\label{alg:max_coup}
\end{algorithm}

\section{Theoretical Considerations}\label{sec:theory}

Most of the assumptions are given in Appendix \ref{app:ass}, where they are also discussed. The final two assumptions are
in Appendices \ref{app:conv} and \ref{app:bias} respectively. If $\varphi:\mathsf{X}\rightarrow\mathbb{R}$, we denote by
$\textrm{Lip}(\mathsf{X})$ the collection of Lipschitz functions, that is that there exists a $C<+\infty$ such
that for any $(x,\overline{x})\in\mathsf{X}^2$
$$
|\varphi(x)-\varphi(\overline{x})| \leq C|x-\overline{x}|
$$
and we recall $|\cdot|$ is the $\mathbb{L}_1$-norm. 

\begin{theorem}\label{theo:main_result}
Assume (A\ref{ass:1}-\ref{ass:4}).  Then for any $(\varphi,p)\in\mathcal{B}_b(\mathsf{X})\cap\textrm{\emph{Lip}}(\mathsf{X})\times\{1,\dots,T\}$ there exists a $C<+\infty$ such that for any 
$\epsilon>0$, $L=\mathcal{O}(|\log(\epsilon)|)$, $L\in\mathbb{N}$, there exists a $(S_2^{\epsilon},\dots,S_L^{\epsilon})\in\mathbb{N}^{L-1}$ and any $S_1^{\epsilon}\in\mathbb{N}$ so that:
$$
\mathbb{\overline{E}}\left[\left(\pi_p^{S_{1:L}^{\epsilon}}(\varphi)-\pi_p(\varphi)\right)^2\right] \leq C\left(
\sum_{l=1}^{L}\frac{\Delta_l}{S_l^{\epsilon}} + \epsilon^2
\right).
$$
\end{theorem}

\begin{proof}
Throughout the proof $C$ is a constant that does not depend upon $l$ and whose value can change from line-to-line.
Using first the $C_2$-inequality twice we have the upper-bound:
$$
\mathbb{\overline{E}}\left[\left(\pi_p^{S_{1:L}^{\epsilon}}(\varphi)-\pi_p(\varphi)\right)^2\right] \leq 
C\Bigg((\pi_p^L(\varphi)-\pi(\varphi))^2 +
$$
$$
\mathbb{\overline{E}}\left[\left(\pi_p^{S_1^{\epsilon},1}(\varphi)-\pi_p^1(\varphi)\right)^2\right] + 
\mathbb{\overline{E}}\left[\left(
\sum_{l=2}^L\left\{\pi_p^{S_l^{\epsilon},l}(\varphi) - \overline{\pi}_p^{S_l^{\epsilon},l-1}(\varphi)\right\} -
\sum_{l=2}^L\left\{\pi_p^l(\varphi)-\pi_p^{l-1}(\varphi)\right\}
\right)^2\right]
\Bigg).
$$
For the first two terms on the R.H.S.~of the inequality one can use Lemma \ref{lem:8} in the Appendix and standard results
for particle filters (e.g.~\cite[Lemma A.3.]{beskos}) to obtain the upper-bound:
$$
C\Bigg(\Delta_L^2 + \frac{\Delta_1}{S_1^{\epsilon}} + \mathbb{\overline{E}}\left[\left(\pi_p^{S_1^{\epsilon},1}(\varphi)-\pi_p^1(\varphi)\right)^2\right] + 
\mathbb{\overline{E}}\left[\left(
\sum_{l=2}^L\left\{\pi_p^{S_l^{\epsilon},l}(\varphi) - \overline{\pi}_p^{S_l^{\epsilon},l-1}(\varphi)\right\} -
\sum_{l=2}^L\left\{\pi_p^l(\varphi)-\pi_p^{l-1}(\varphi)\right\}
\right)^2\right]\Bigg).
$$
For the last term in the above expression, one can multiply out the brackets and apply Lemmata \ref{lem:6}-\ref{lem:7}. 
The proof is completed by noting the specification of $L$ in the statement amd
 specifying $(S_2^{\epsilon},\dots,S_{L}^{\epsilon})$ large enough so that the sum of the upper-bounds are $\mathcal{O}(\epsilon^2)$.
\end{proof}

The result that is given in Theorem \ref{theo:main_result} is not the usual bound that is given for MSE bounds in the multilevel Monte Carlo literature. The reason for this, is the difficulty of analyzing Algorithm \ref{alg:mlpf} and is explained extensively in Appendix \ref{app:conv}. The assumption (A\ref{ass:5}) in Appendix \ref{app:conv} is essentially a result one would want to prove; therefore one can say that our analysis needs to be extended. None-the-less the result is indicative as it suggests that
one should choose $S_l^{\epsilon}=\mathcal{O}(\epsilon^{-2}\Delta_lL)$ to achieve a MSE of $\mathcal{O}(\epsilon^2)$ (to choose $L$ one can use the bias result, Lemma \ref{lem:8}). If the cost of simulation of Algorithm \ref{alg:mlpf} is $\mathcal{O}(\Delta_l^{-1})$, which is not totally clear due to the random amount of time one must apply the solver, one recovers the classical complexity $\mathcal{O}(\epsilon^{-2}\log(\epsilon)^2)$ to achieve a MSE of
$\mathcal{O}(\epsilon^2)$. Even though our mathematical results need extension, we choose  $S_l^{\epsilon}=\mathcal{O}(\epsilon^{-2}\Delta_lL)$ and investigate the complexity numerically. 

We remark if one uses a single particle filter at level $L$
then one expects a cost of $\mathcal{O}(\epsilon^{-3})$ cost to achieve a $\mathcal{O}(\epsilon^2)$ MSE (assuming the cost of simulation of Algorithm \ref{alg:pf} is $\mathcal{O}(\Delta_l^{-1})$, when using discretization $\Delta_l$ instead of $\Delta_1$). 
This follows from standard results on particle filters (e.g.~\cite[Lemma A.3.]{beskos}), which imply that $S_L=\mathcal{O}(\epsilon^{-2})$  and the bias result, Lemma \ref{lem:8} in the Appendix, which says that $L=\mathcal{O}(|\log(\epsilon)|)$. 
We give a numerical comparison using a single particle filter versus using multilevel particle filters in the next section.

\section{Numerical Results}\label{sec:numerics}

In this section, we provide numerical illustrations of our methodology for the Multilevel Particle Filters (MLFP) algorithm, applied to PDMPs while comparing it to the standard Particle Filter. We'll outline and test our algorithms with different neuroscience models, where we present and demonstrate the benefit of using MLPF.

\subsection{Biological Neuron Models}

The Hodgkin--Huxley (HH) model (\cite{HH,Erment}) is one of the most well-known models in neuroscience, and is a family of conductance-based neurons. It is built to describe the behavior of the squid's neuron, in particular, for the propagation of an action potential (spike) along the neuron. Hodgkin and Huxley were able to perform experiments on space-clamped squid's giant axon, due to the fact that the diameter of this axon is greater than others, where their model is equivalent to an electrical analogy. Mathematically, the HH model is a system of first-order nonlinear ordinary differential equations with four coupled equations:

\begin{equation}
\left\{\begin{array}{lllclcll}
 \dfrac{dV}{dt}= \dfrac{1}{C_{m}} \left[ -\overline{g_{Na}} m^3 h (V-E_{Na})-\overline{g_K} n^4 (V-E_K)-\overline{g_L}(V-E_{L}) + I_{\text{ext}}\right],
 \vspace{1mm}\\
\dfrac{dm}{dt}=\alpha_{m}(V)(1-m)-\beta_{m}(V)m, \vspace{1mm}\\
\dfrac{dh}{dt}=\alpha_{h}(V)(1-h)-\beta_{h}(V)h,\vspace{1mm}\\
\dfrac{dn}{dt}=\alpha_{n}(V)(1-n)-\beta_{n}(V)n.

\end{array}\right.
\label{moneq1}
\end{equation}

Where the variable\ $V(t)$\ describes the membrane potential of the cell, meaning that it is the difference between intracellular and extracellular at time\ $t$. The variables\ $m(t)$,\ $h(t)$\ and\ $n(t)$\ represent the probability that these channels gates be open at time\ $t$\ depending on the membrane potential. Furthermore\ $m(t)$\ (resp.\ $h(t)$)\ is the activation (resp. inactivation) of the sodium flow current while\ $n(t)$\ is the activation of the potassium flow current. The constant\ $\overline{g_{Na}}$,\ (resp.\ $\overline{g_K}$)\ is the maximal conductances of a sodium (resp. potassium) channel, and\ $\overline{g_L}$\ is the leaky conductance. The parameters\ $E_{Na}$,\ $E_{K}$, and\ $E_{L}$\ are called reversal potentials and are obtained by Nernst's equation. The functions\ $\alpha_{x}(V)$\ and\ $\beta_{x}(V)$, where\ $x=m,h,n$\ describe the transfer rate between the opening states and the closing states.\ $C_{m}$\ is the membrane capacitance and the applied current\ $I_{\text{ext}}$\ models stimulating external drive. The formulas are given below,  \\

\begin{flushleft}

$\alpha_m(V):= \frac{0.1\, ( V+40)}{1-\exp\big(\frac{-(V+40)}{10}\big)},~~~~~~~~~\beta_m(V):=4\exp\Big(\frac{-(V+65)}{18}\Big),~~~~~~~~~\alpha_h(V):=0.07\exp\Big(\frac{-(V+65)}{20}\Big)$,\\

$\beta_h(V):=\frac{1}{1+\exp\big(\frac{-(V+35)}{10}\big)},~~~~~~~~~\alpha_n(V):=\frac{0.01\, (V+55)}{1-\exp\big(\frac{-(V+55)}{10}\big)},~~~~~~~~~\beta_n(V):=0.125 \exp\Big(\frac{-(V+65)}{80}\Big)$,\\

\end{flushleft}

$E_{K} = -77\; \text{mV}$,\quad $E_{Na} = 50\;\text{mV}$,\quad $E_{L} = -54.4\;\text{mV}$,\quad $\overline{g}_{K} =36\; mS/cm^{2}$,\\
 $\overline{g}_{Na} =120\; mS/cm^{2}$,\quad  $\overline{g}_{L} =0.3\; mS/cm^{2} $,\quad $C_m = 1 \mu F/cm^2$.

\subsubsection{The 2D Morris--Lecar Model}

One of the simplest models in computational neuroscience, for the production of spikes, is a reduced 2-variable model proposed by Kathleen Morris and Harold Lecar (ML) (\cite{Morris--Lecar}) . The ML equations are simpler than the HH equations, where they can easily explain the dynamics of the barnacle muscle fiber by exploiting bifurcation theory, and that they exhibit many features of neuronal activity (e.g., firing events, emergent dynamics...). 

The ML model has three ion channels: a potassium channel, a leak, and a calcium channel. In the simplest version of the model, the calcium current depends instantaneously on the voltage.

Mathematically, the ML model is a system of first-order nonlinear ordinary differential equations with two coupled equations as follows:
\begin{equation}
\left\{\begin{array}{ll}
 \dfrac{dV}{dt}&= \quad
 \dfrac{1}{C_{m}} \left[ -\overline{g_{Ca}} M_{\infty} (V-E_{Ca})-\overline{g_K} n (V-E_K)-\overline{g_L}(V-E_{L}) + I_{\text{ext}}\right]\vspace{2mm}\\
\dfrac{dn}{dt}&=\quad
\alpha_{n}(V)(1-n)-\beta_{n}(V)n,

\end{array}\right.
\label{ML}
\end{equation}

where,
\begin{align*}
M_{\infty}&= \frac{1}{2} (1 + \tanh ( (V - V_1)/V_2 ) ),\\
N_{\infty}&= \frac{1}{2} (1 + \tanh ( (V - V_3)/V_4 ) )
\end{align*}
and
\begin{align*}
\alpha_{n}(V) &= \lambda_{n}(V) N_{\infty}(V), \\
\beta_{n}(V) &= \lambda_{n}(V) (1 - N_{\infty})(V),\\
\lambda_{n}(V) &= \overline{\lambda}_{n} \cosh ( (V - V_3)/2V_4 ).
\end{align*}

Here, $V_{1,2,3,4}$ are parameters chosen to fit voltage clamp data. Hence, we consider the PDMP $(X_t)_{t \geq 0}$ associated with the equations (\ref{ML}) as follows

\begin{align*}
f(u,V) &= \frac{1}{C_{m}} \left[ -\overline{g_{Ca}} M_{\infty} (V-E_{Ca})-\overline{g_K} \frac{u}{N_n} (V-E_K)-\overline{g_L}(V-E_{L}) + I_{\text{ext}}\right],\\
\lambda (u, V) &= (N_n - u)\alpha_{n}(V) + u \beta_{n}(V),\\
   Q((u,V), \left\{( u+1,V )\right\}) &= \frac{(N_{\infty} - u) \alpha_{n}(V)}{\lambda (u, V)},\\
   Q((u,V), \left\{ (u - 1,V) \right\}) &= \frac{ u \beta_{n}(V)}{\lambda (u, V)}.
\end{align*}
The observation data $Y_k$ that we choose is  
$Y_k\mid(V_{k\delta}, n_{k\delta} )\sim \mathcal{N}( V_{k\delta} ,\tau^2)$ where $\delta=0.5$ and $\tau^2=0.1$

In the following table, we set the value of the parameters we used in the simulations.

\begin{table}[H]
\begin{center}
\begin{tabular}{ ||c| c| c|  c|| } 
\hline \hline

$V(t_0) = -20\, \text{mV}$ & $n(t_0) = 0 $   & $E_{K} = -84\; \text{mV} $  & $E_{L} = -60\;\text{mV}$ \\
\hline \hline

 $E_{Ca} = 120\;\text{mV}$ & $\overline{g}_{K} =8\; mS/cm^{2}$ & $\overline{g}_{L} = 2\; mS/cm^{2} $ & $\overline{g}_{Ca} = 4.4\; mS/cm^{2}$
 \\
 \hline \hline
  $C_m = 20 \mu F/cm^2$ & $\overline{\lambda}_{n} = 0.04 $ & $I_{\text{ext}}  = 100  \mu A/cm^2 $ & $N_n = 100$ 
\\
\hline

$V_1 = -1.2$ & $V_2 = 18$ & $V_3 = 2$ & $V_4 = 30$ 
\\
\hline
 \hline
\end{tabular}
\caption{Parameter choices of the ML model.}
\label{tab:tab1}
\end{center}
\end{table}

\subsubsection{ $I_K + I_L$ Model}

We also consider an alternative model, which is a conductance-based neuron model. This model was suggested as a two-dimensional (2D) simplification of the realistic HH model. Indeed, this 2D neuron model consists of two coupled equations with two main variables $V$ and $m_K$. It is governed by the following pair of differential equations

\begin{equation}
\left\{\begin{array}{ll}

 \dfrac{dV}{dt}&=\quad  -\overline{g_K} m_K (V-E_K)-\overline{g_L}(V-E_{L}) + I_{\text{ext}} \vspace{2mm}\\
\dfrac{dm_K}{dt}&=\quad
\alpha_{m_K}(V)(1-m_K)-\beta_{m_K}(V)m_K,

\end{array}\right.
\label{CBN}
\end{equation}
where $\alpha_{m_K}(V) = \frac{0.1\, ( V+40)}{1-\exp\big(\frac{-(V+40)}{10}\big)}$ and $\beta_{m_K}(V)=4\exp\big(\frac{-(V+65)}{18}\big)$. We consider the PDMP $(X_t)_{t \geq 0}$ associated with the equations (\ref{CBN}) as follows

\begin{align*}    
f(u,V) &=  -\overline{g_K} \frac{u}{N_{m_K}} (V-E_K)-\overline{g_L}(V-E_{L}) + I_{\text{ext}},\\
\lambda (u, V) &= (N_{m_K} - u)\alpha_{m_K}(V) + u \beta_{m_K}(V),\\
Q((u,V), \left\{ (u+1,V) \right\}) &= \frac{(N_{m_K} - u) \alpha_{m_K}(V)}{\lambda (u, V)},\\
Q((u,V), \left\{ (u - 1,V) \right\}) &= \frac{ u \beta_{m_K}(V)}{\lambda (u, V)}.
\end{align*}
The observation data $Y_k$ that we choose is  
$Y_k\mid(V_{k\delta}, n_{k\delta} )\sim \mathcal{N}( V_{k\delta} ,\tau^2)$ where $\delta=0.5$ and $\tau^2=0.2$

In the following table, we set the value of the parameters we used in the simulations.

\begin{table}[H]
\begin{center}
\begin{tabular}{ ||c| c| c|  c|| } 
\hline \hline

$V(t_0) = -65\, \text{mV}$ & $ m_{K}(t_0) = 0 $   & $E_{K} = -77\; \text{mV} $  & $E_{L} = -54.4\;\text{mV}$ \\
\hline \hline

 $I_{\text{ext}}  = 10 \sin (\frac{\pi t}{10})$  & $\overline{g}_{K} =36\; mS/cm^{2}$ & $\overline{g}_{L} = 0.3\; mS/cm^{2} $ & $ N_{m_K} = 100$
 \\

\hline

\hline
 \hline
\end{tabular}
\caption{The value of the parameters of the $I_K + I_L$ Model.}
\label{tab}
\end{center}
\end{table}

\subsection{Simulation Settings and Results}
For our numerical experiments, with neuroscience models above, we consider multilevel estimators at levels $l = \{3,4,5,6,7\}$. The results are generated from a high-resolution simulation of particle filters to approximate the ground truth of our models. \\
The resampling is done adaptively. For the particle filters, resampling is done when the effective sample size (ESS) is less than $1/2$ of the particle numbers. For the coupled filters, we use the ESS of the coarse filter as the measurement of discrepancy. Each simulation is repeated $100$ times.


We now present our numerical simulations to show the benefits of applying the MLPF algorithm to the PDMPs models, in comparison to the PF. Our results compare the MSE directly with the cost, which considers the rate through each targeted MSE $\mathbb{E}[\phi(V_{n \delta} | y_{1:n})]$, which are shown in Figure \ref{fig:MSEvsCost}. The figure shows that as we increase the levels from $l=3$ to $l=7$, the difference in the cost between the methods also increases. These figures show the advantage and accuracy of using MLPF to PDMP systems. Table 2 presents the estimated rates of MSE with respect to cost. This agrees with our theory, which predicts a  complexity rate of $\mathcal{O}(\epsilon^{-3})$ for the particle filter and $\mathcal{O}(\epsilon^{-2} \log (\epsilon)^2)$ for the multilevel particle filter of PDMPs.

\begin{figure}
\centering
\subfigure{\includegraphics[width=14cm,height=5cm]{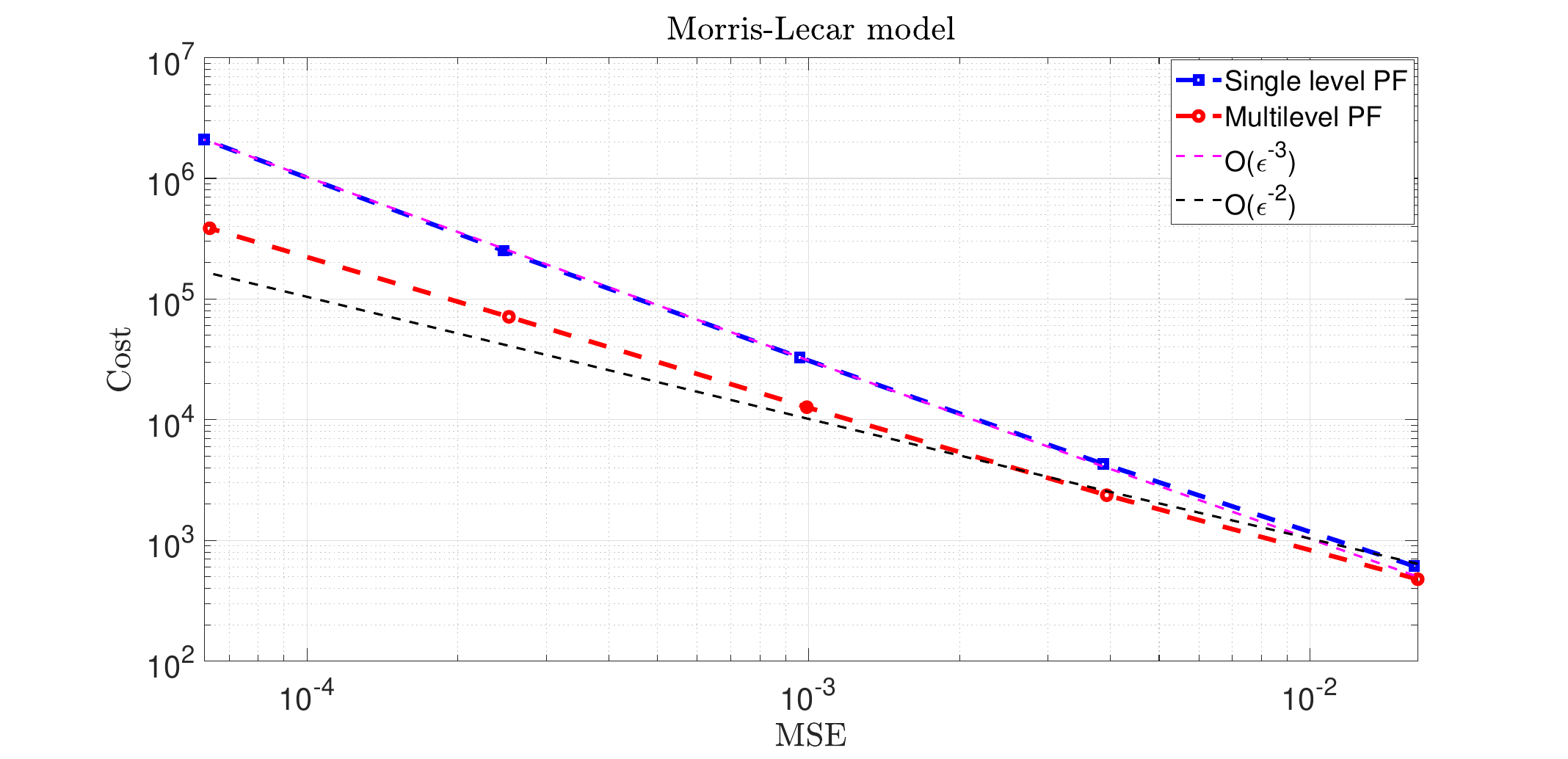}}
\subfigure{\includegraphics[width=14cm,height=5cm]{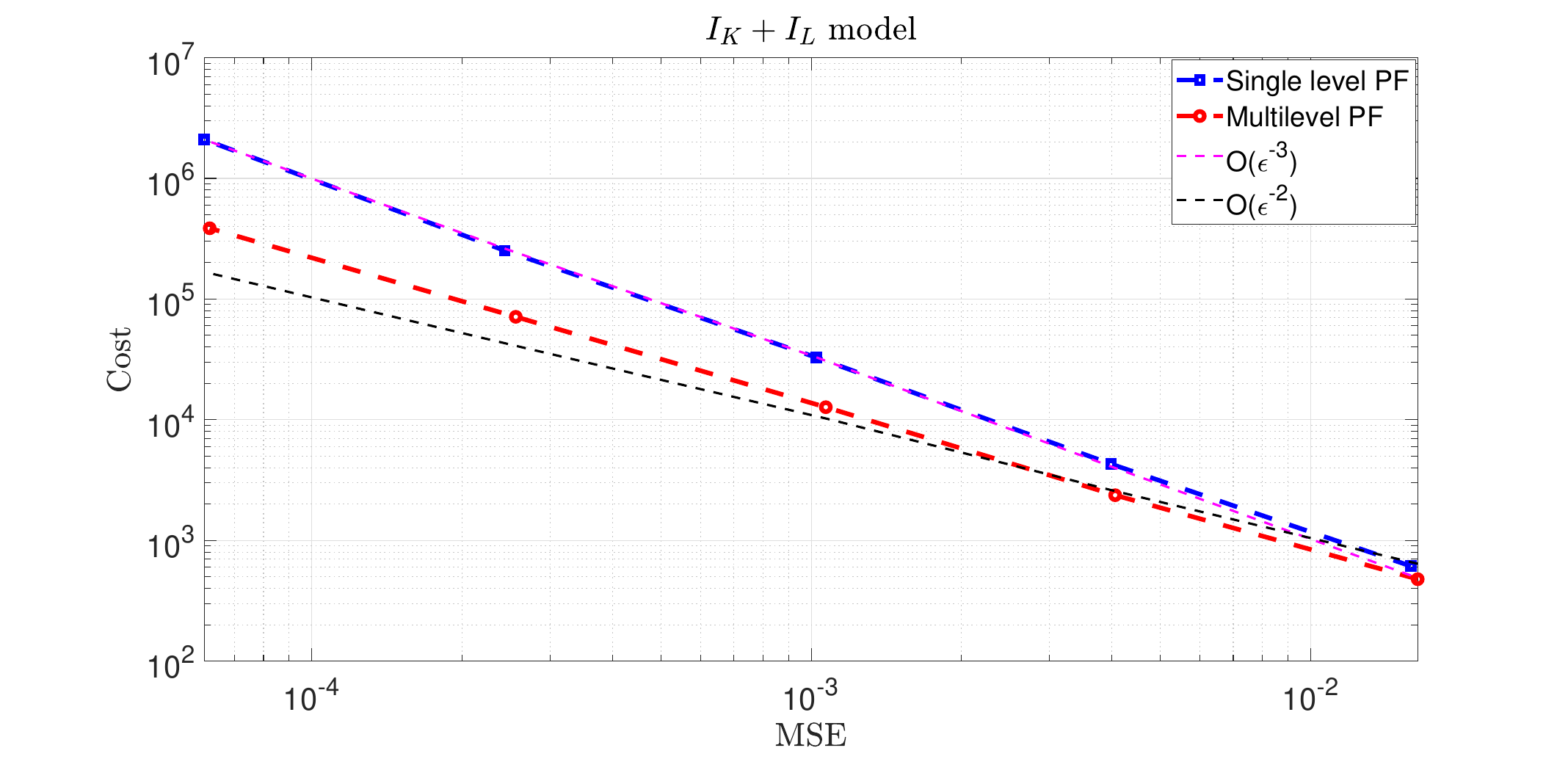}}
 \caption{Cost rates as a function of the mean squared error.}
    \label{fig:MSEvsCost}
\end{figure}

\begin{table}[H]
\begin{center}
\begin{tabular}{ c c c  } 
\hline \hline
Model  & Particle Filter & Multilevel Particle Filter   \\
\hline
\hline
\textbf{$I_K + I_L$ Model} & -1.46 &-1.2   \\ 
 
\textbf{Morris--Lecar system} & -1.47 &-1.19  \\

 \hline
\end{tabular}
\caption{Estimated rates of MSE with respect to cost.}
\label{tab:tab2}
\end{center}
\end{table}

\subsubsection*{Acknowledgements}

AJ \& MM were supported by KAUST baseline funding.

\appendix

\section{Proofs}

This appendix consists of four sections and should be read in order; the logical progression of the proofs is such that it does not make sense to read it in another way. All of the results have been written in order to prove Theorem \ref{theo:main_result}.
In Section \ref{app:ass} we give most of the assumptions that are employed in our analysis, along with a discussion of them.
We note that two of the assumptions can be found in later sections, where they are explicitly used. In Section \ref{app:rate} we give several results which can be used to obtain the `rates' i.e.~the term $\Delta_l$ in the upper-bound in Theorem \ref{theo:main_result}. In Section \ref{app:conv} we explain why our main result needs to be extended in terms of mathematical complexity and give the main technical results associated to an assumption we make. Finally in Section \ref{app:bias} we prove that our discretized filter converges to the exact filter and specify exactly the rate as a function of $\Delta_l$.

\subsection{Assumptions}\label{app:ass}

Throughout the appendix $C$ is a constant that does not depend upon $l$ and whose value can change from line-to-line.

\begin{hypA}\label{ass:1}
 For $0\le s\le t\le T$,  there exists a $C<+\infty$ such that for any $(l,x,\overline{x})\in\mathbb{N}\times\mathsf{X}^2$:
$$
|\Phi(x,t)-\Phi^l_{[s,t]}(\overline{x},t)| \leq C\left(|x-\overline{x}| + \Delta_l\right).
$$
\end{hypA}

\begin{hypA}\label{ass:2}
We have that
\begin{enumerate}
\item{There exist $0<\underline{\lambda}<\overline{\lambda}<\lambda^{\star}$ such that for every $x\in\mathsf{X}$:
$$
\underline{\lambda}\leq \lambda(x)\leq \overline{\lambda}.
$$
}
\item{There exist $0<\underline{C}<\overline{C}<+\infty$ such that for every $(x,u)\in\mathsf{X}\times\mathsf{U}$:
$$
\underline{C}\leq Q(x,u)\leq \overline{C}.
$$
}
\item{For any $y\in\mathsf{Y}$ there exist $0<\underline{C}<\overline{C}<+\infty$ such that for every $x\in\mathsf{X}$:
$$
\underline{C}\leq g(x,y)\leq \overline{C}.
$$
}
\item{For any $y\in\mathsf{Y}$, $g(\cdot,y)\in\textrm{Lip}(\mathsf{X})$.}
\end{enumerate}
\end{hypA}

\begin{hypA}\label{ass:3}
For any $(\varphi,p,\kappa)\in\mathcal{B}_b(\mathsf{X})\cap\textrm{Lip}(\mathsf{X})\times\{1,\dots,T\}\times(0,\infty)$ there exists a $C<+\infty$ such that for any $(l,x,\overline{x})\in\mathbb{N}\times\mathsf{X}^2$:
$$
\mathbb{E}^l\left[\left|\varphi(X_p^l)-\varphi(X_p^{l-1})R_{p-1,p}^l(\mathbf{X}_{p-1}^l,\Xi_{p-1,p}^l,\mathbf{X}_{p}^l)\right|^{\kappa}|\mathbf{X}_{p-1}^l=(x,\overline{x})\right] \leq C\left(|x-\overline{x}|^{\kappa}+\Delta_l^{\kappa}\right).
$$
\end{hypA}

Assumption (A\ref{ass:1}) is a reasonable assumption that is satisfied for many discretization schemes, including the Euler approximation method used in this article. Assumption (A\ref{ass:2}) in terms of $\lambda$, is quite reasonable as we start with the premise that $\lambda$ is bounded; that it is further lower and upper-bounded does not seem overly restrictive in general.
The assumptions on $Q$ and $g$ are ones that have been adopted in many analyses of multilevel particle filters e.g.~\cite{coup_clt,mlpf,anti_pf} and essentially imply that the associated spaces are compact. Generally, weakening these assumptions needs the application of drift conditions \cite{delm,coup_clt} and leads to longer proofs; hence we do not do this.
For (A\ref{ass:3}), on inspection of \cite[Theorem 3.2.]{lemaire} it appears the rate in terms of $\Delta_l$ is appropriate. The continuity in terms of $(x,\overline{x})$ also seems reasonable given (A\ref{ass:1}).

\subsection{Rate Proofs}\label{app:rate}

\begin{lem}\label{lem:1}
Assume (A\ref{ass:1}).  Then for any $(p,\kappa)\in\{1,\dots,T\}\times(0,\infty)$ there exists a $C<+\infty$ such that for any $(l,x,\overline{x})\in\{2,3,\dots\}\times\mathsf{X}^2$:
$$
\mathbb{E}^l\left[\left|X_p^l-X_p^{l-1}\right|^{\kappa}|\mathbf{X}_{p-1}^l=(x,\overline{x})\right] \leq C\left(|x-\overline{x}|^{\kappa}+\Delta_l^{\kappa}\right).
$$
\end{lem}

\begin{proof}
First we consider $N_{[p-1,p]}^l=0$ then clearly
$$
\mathbb{I}_{\{N_{[p-1,p]}^l=0\}}\left|X_p^l-X_p^{l-1}\right| = \mathbb{I}_{\{N_{[p-1,p]}^l=0\}}\left|\Phi_{[p-1,p]}^l(x,1)-\Phi_{[p-1,p]}^{l-1}(\overline{x},1)\right|.
$$
Using (A\ref{ass:1}), clearly, we have
$$
\mathbb{I}_{\{N_{[p-1,p]}^l=0\}}\left|X_p^l-X_p^{l-1}\right|^{\kappa} \leq C\left(|x-\overline{x}|^{\kappa}+\Delta_l^{\kappa}\right).
$$
Using a simple recursive argument, for any $N_{[p-1,p]}^l$ we would have the almost sure upper-bound
$$
\left|X_p^l-X_p^{l-1}\right|^{\kappa} \leq C^{N_{[p-1,p]}^{l}}\left(|x-\overline{x}|^{\kappa}+\Delta_l^{\kappa}\right).
$$
Thus, noting that  almost surely $N_{[p-1,p]}^{l}\leq N_{[p-1,p]}^{\star,l}$ and that the latter random variable has exponential moments, the proof is concluded.
\end{proof}

 At any time point, $k$ of Algorithm  \ref{alg:mlpf} we will denote the resampled index of particle $i\in\{1,\dots,S_l\}$ as $I_k^{i,l}$ (level $l$) and $I_k^{i,l-1}$ (level $l-1$). Now let $I_k^l(i)=I_k^{i,l}$, $I_k^{l-1}(i)=I_k^{i,l-1}$ and define $\mathsf{S}_k^l$ the collection of indices that choose the same ancestor at each resampling step, i.e.
\begin{eqnarray*}
\mathsf{S}_k^l & = & \{i\in\{1,\dots,S_l\}:I_k^l(i)=I_k^{l-1}(i),I_{k-1}^l\circ I_k^l(i)=I_{k-1}^{l-1}\circ I_k^{l-1}(i),\dots,\\ & & I_{1}^l\circ I_2^l\circ\cdots\circ I_k^l(i) = I_{1}^{l-1}\circ I_2^{l-1}\circ\cdots\circ I_k^{l-1}(i)\}.
\end{eqnarray*}
We use the convention that $\mathsf{S}_0^l=\{1,\dots,S_l\}$. For $(a,b)\in\mathbb{R}$, $a\wedge b$ is the minimum of $a$ and $b$. If $A$ is a finite set, we denote by $\textrm{Card}(A)$ its cardinality.

\begin{lem}\label{lem:2}
Assume (A\ref{ass:1}-\ref{ass:2}).  Then for any $(p,\kappa)\in\{1,\dots,T\}$ there exists a $C<+\infty$ such that for any $(l,S_l)\in\{2,3,\dots\}\times\mathbb{N}$:
$$
\mathbb{\overline{E}}\left[\frac{1}{S_l}\sum_{i\in\mathsf{S}_{p-1}}|X_p^{i,l}-\overline{X}_{p}^{i,l-1}|^{\kappa}\right]
\leq C\Delta_l^{\kappa}.
$$
\end{lem}

\begin{proof}
The proof is very similar to that of \cite[Lemma D.3.]{mlpf} and as such, we skip some steps for brevity; the full details can be followed from the afore-mentioned proof. The case $p=1$ follows via Lemma \ref{lem:1}, so we shall assume that the result holds at a rank $p-1$, $p\geq 1$ and consider time $p$. Conditioning on the resampled particles and applying \cite[Lemma D.3.]{mlpf} gives the upper-bound
$$
\mathbb{\overline{E}}\left[\frac{1}{S_l}\sum_{i\in\mathsf{S}_{p-1}^l}\left|X_p^{i,l}-\overline{X}_{p}^{i,l-1}\right|^{\kappa}\right] \leq
C\left(\Delta_l^2 + \mathbb{\overline{E}}\left[\frac{1}{S_l}\sum_{i\in\mathsf{S}_{p-1}^l}\left|X_{p-1}^{I_{p-1}^{i,l},l}-\overline{X}_{p-1}^{I_{p-1}^{i,l-1},l-1}\right|^{\kappa}\right]\right).
$$
Then using the same logic as \cite[Lemma D.3.]{mlpf}, we have that
$$
 \mathbb{\overline{E}}\left[\frac{1}{S_l}\sum_{i\in\mathsf{S}_{p-1}^l}\left|X_{p-1}^{I_{p-1}^{i,l},l}-\overline{X}_{p-1}^{I_{p-1}^{i,l-1},l-1}\right|^{\kappa}\right] = 
  \mathbb{\overline{E}}\left[\frac{\textrm{Card}(\mathsf{S}_{p-1}^l)}{S_l}
  \frac{
\sum_{i\in\mathsf{S}_{p-2}^l}\left|X_{p-1}^{i,l}-\overline{X}_{p-1}^{i,l-1}\right|^{\kappa}\Big\{
G_{p-1}^{i,l}\wedge \overline{G}_{p-1}^{i,l-1}
\Big\}
  }{
  \sum_{i\in\mathsf{S}_{p-2}^l}
\Big\{G_{p-1}^{i,l}\wedge \overline{G}_{p-1}^{i,l-1}\Big\}
  }
  \right]
$$
Now, using (A\ref{ass:2}), one can find constants $0<\underline{c}<1<\overline{c}<+\infty$ with $\overline{c}>1/\underline{c}$ independent of $l$ so that almost surely
\begin{equation}\label{r:bound}
\underline{c}^{N_{[p-1,p]}^{\star,l}} \leq R_{p-1,p}^l(\mathbf{x}_{p-1},\Xi^l_{p-1,p},\mathbf{x}_{p}) \leq \overline{c}^{N_{[p-1,p]}^{\star,l}}
\end{equation}
where $N_{[p-1,p]}^{i,\star,l}$ are the number of events of the Poisson process for the $i^{\textrm{th}}-$sample.
Therefore, we have that almost surely
\begin{equation}\label{g:bound}
\frac{\underline{C}\underline{c}^{N_{[p-1,p]}^{i,\star,l}}}{\overline{C}\sum_{j=1}^{S_l} \overline{c}^{N_{[p-1,p]}^{j,\star,l}}} \leq \overline{G}_{p-1}^{i,l-1} \leq
\frac{\overline{C}\overline{c}^{N_{[p-1,p]}^{i,\star,l}}}{\underline{C}\sum_{j=1}^{S_l} \underline{c}^{N_{[p-1,p]}^{j,\star,l}}}.
\end{equation}
and hence that for some $C<+\infty$
$$
 \mathbb{\overline{E}}\left[\frac{1}{S_l}\sum_{i\in\mathsf{S}_{p-1}^l}\left|X_{p-1}^{I_{p-1}^{i,l},l}-\overline{X}_{p-1}^{I_{p-1}^{i,l-1},l-1}\right|^{\kappa}\right]  \leq 
 $$
 $$
 C
 \mathbb{\overline{E}}\left[\frac{\textrm{Card}(\mathsf{S}_{p-1}^l)}{S_l}
  \frac{
\sum_{i\in\mathsf{S}_{p-2}^l}\left|X_{p-1}^{i,l}-\overline{X}_{p-1}^{i,l-1}\right|^{\kappa} \frac{\overline{c}^{N_{[p-1,p]}^{i,\star,l}}}{\sum_{j=1}^N \underline{c}^{N_{[p-1,p]}^{j,\star,l}}}
  }{
  \sum_{i\in\mathsf{S}_{p-2}^l} \frac{\underline{c}^{N_{[p-1,p]}^{i,\star,l}}}{\sum_{j=1}^N \overline{c}^{N_{[p-1,p]}^{j,\star,l}}}
  }
  \right] 
   =  
   $$
   $$
  C
 \mathbb{\overline{E}}\left[\frac{\textrm{Card}(\mathsf{S}_{p-1}^l)}{S_l}
   \frac{\sum_{(i,j)\in\mathsf{S}_{p-2}^l\times[S_l]}
 |X_{p-1}^{i,l}-\overline{X}_{p-1}^{i,l-1}|^{\kappa}\overline{c}^{N_{[p-1,p]}^{i,\star,l}+N_{[p-1,p]}^{j,\star,l}}}
 {\sum_{(i,j)\in\mathsf{S}_{p-2}^l\times[S_l]}
\underline{c}^{N_{[p-1,p]}^{i,\star,l}+N_{[p-1,p]}^{j,\star,l}}}
   \right]
$$
where we use the short-hand $[S_l]=\{1,\dots,S_l\}$.
By construction $\underline{c}\overline{c}\geq 1$, so one has
$$
 \mathbb{\overline{E}}\left[\frac{1}{S_l}\sum_{i\in\mathsf{S}_{p-1}^l}|X_{p-1}^{I_{p-1}^{i,l},l}-\overline{X}_{p-1}^{I_{p-1}^{i,l-1},l-1}|^{\kappa}\right]  \leq 
 $$
$$
C\mathbb{\overline{E}}\left[\frac{\textrm{Card}(\mathsf{S}_{p-1}^l)}{S_l}
   \frac{\sum_{(i,j)\in\mathsf{S}_{p-2}^l\times[S_l]}
 |X_{p-1}^{i,l}-\overline{X}_{p-1}^{i,l-1}|^{\kappa}\overline{c}^{N_{[p-1,p]}^{i,\star,l}+N_{[p-1,p]}^{j,\star,l}}}
 {\sum_{(i,j)\in\mathsf{S}_{p-2}^l\times[S_l]}
\overline{c}^{-N_{[p-1,p]}^{i,\star,l}-N_{[p-1,p]}^{j,\star,l}}}
   \right].
$$
Conditioning on $\mathsf{S}_{p-1}^l$ and filtration generated by the particle system (Algorithm \ref{alg:mlpf}) up-to time $p-1$, one can use the conditional Jensen inequality to arrive at
\begin{eqnarray*}
 \mathbb{\overline{E}}\left[\frac{1}{S_l}\sum_{i\in\mathsf{S}_{p-1}^l}|X_{p-1}^{I_{p-1}^{i,l},l}-\overline{X}_{p-1}^{I_{p-1}^{i,l-1},l-1}|^{\kappa}\right]  & \leq &
C\mathbb{\overline{E}}\Bigg[\frac{\textrm{Card}(\mathsf{S}_{p-1}^l)}{S_l}
\Big(\tfrac{1}{S_l\textrm{Card}(\mathsf{S}_{p-2}^l)}\sum_{(i,j)\in\mathsf{S}_{p-2}^l\times[S_l]}
 |X_{p-1}^{i,l}-\overline{X}_{p-1}^{i,l-1}|^{\kappa}\times \\ & &\overline{c}^{N_{[p-1,p]}^{i,\star,l}+N_{[p-1,p]}^{j,\star,l}}\Big)\Big(
\tfrac{1}{S_l\textrm{Card}(\mathsf{S}_{p-2}^l)} \sum_{(i,j)\in\mathsf{S}_{p-2}^l\times[S_l]}
\overline{c}^{N_{[p-1,p]}^{i,\star,l}+N_{[p-1,p]}^{j,\star,l}}\Big)
   \Bigg].
\end{eqnarray*}
Using the independence of the $N_{[p-1,p]}^{j,\star,l}$, one easily deduces the upper-bound
$$
 \mathbb{\overline{E}}\left[\frac{1}{S_l}\sum_{i\in\mathsf{S}_{p-1}^l}|X_{p-1}^{I_{p-1}^{i,l},l}-\overline{X}_{p-1}^{I_{p-1}^{i,l-1},l-1}|^{\kappa}\right]   \leq
 C \mathbb{\overline{E}}\Bigg[\frac{\textrm{Card}(\mathsf{S}_{p-1}^l)}{S_l}
 \frac{1}{\textrm{Card}(\mathsf{S}_{p-2}^l)}\sum_{i\in\mathsf{S}_{p-2}^l} |X_{p-1}^{i,l}-\overline{X}_{p-1}^{i,l-1}|^{\kappa}
    \Bigg]
$$
and as $\textrm{Card}(\mathsf{S}_{p-1}^l)\leq \textrm{Card}(\mathsf{S}_{p-2}^l)$, the proof is concluded by induction.
\end{proof}

\begin{lem}\label{lem:3}
Assume (A\ref{ass:1}-\ref{ass:3}).  Then for any $(\varphi,p,\kappa)\in\mathcal{B}_b(\mathsf{X})\cap\textrm{\emph{Lip}}(\mathsf{X})\times\{1,\dots,T\}\times(0,\infty)$ there exists a $C<+\infty$ such that for any $(l,S_l)\in\{2,3,\dots\}\times\mathbb{N}$:
$$
\mathbb{\overline{E}}\left[\frac{1}{S_l}\sum_{i\in\mathsf{S}_{p-1}}|\varphi(X_p^{i,l})-\varphi(\overline{X}_{p}^{i,l-1})
R_{p-1,p}^l(\mathbf{X}_{p-1}^{i,l},\Xi_{p-1,p}^{i,l},\mathbf{X}_{p}^{i,l})|^{\kappa}\right]
\leq C\Delta_l^{\kappa}.
$$
\end{lem}

\begin{proof}
This result follows by using (A\ref{ass:3}) along with the calculations in Lemma \ref{lem:2}. As these latter calculations are repeated, we omit them for brevity.
\end{proof}

\begin{lem}\label{lem:4}
Assume (A\ref{ass:1}-\ref{ass:3}).  Then for any $p\in\{0,\dots,T\}$ there exists a $C<+\infty$ such that for any $(l,S_l)\in\{2,3,\dots\}\times\mathbb{N}$:
$$
1- \mathbb{\overline{E}}\left[\frac{\textrm{\emph{Card}}(\mathsf{S}_{p})}{S_l}\right]
\leq C\Delta_l.
$$
\end{lem}

\begin{proof}
The calculations of this proof follow \cite[Lemma D.4.]{mlpf} quite closely and so we skip some steps.
Note that the claim is trivially true at rank $p=0$, so we prove the result by induction on $p$.
We note that
\begin{equation}\label{eq:lem1}
1-\sum_{i=1}^{S_l}G_p^{i,l}\wedge\overline{G}_{p}^{i,l-1} = \frac{1}{2}\sum_{i\in\mathsf{S}_{p-1}^l}\big|G_p^{i,l}-\overline{G}_{p}^{i,l-1}\big| + \frac{1}{2}\sum_{i\notin\mathsf{S}_{p-1}^l}\big|G_p^{i,l}-\overline{G}_{p}^{i,l-1}\big|.
\end{equation}
As
\begin{equation}\label{eq:lem2}
1- \mathbb{\overline{E}}\left[\frac{\textrm{Card}(\mathsf{S}_{p})}{S_l}\right] = 
1 - \mathbb{\overline{E}}\left[\sum_{i=1}^{S_l}G_p^{i,l}\wedge\overline{G}_{p}^{i,l-1}\right] + 
\mathbb{\overline{E}}\left[\sum_{i\notin\mathsf{S}_{p-1}^l}\big|G_p^{i,l}-\overline{G}_{p}^{i,l-1}\big|\right]
\end{equation}
we shall focus on upper-bounding the two terms
\begin{eqnarray}
E_1 & = & \mathbb{\overline{E}}\left[\sum_{i\notin\mathsf{S}_{p-1}^l}\big|G_p^{i,l}-\overline{G}_{p}^{i,l-1}\big|\right]\label{eq:lem3}\\
E_2 & = & \mathbb{\overline{E}}\left[\sum_{i\in\mathsf{S}_{p-1}^l}\big|G_p^{i,l}-\overline{G}_{p}^{i,l-1}\big|\right]
\label{eq:lem4}
\end{eqnarray}
respectively.

For the term $E_1$, we note that by using the triangular inequality, along with $G_p^{i,l}\leq\tfrac{C}{S_l}$ (by (A\ref{ass:2})) and then using the upper-bound in 
\eqref{g:bound}, we have
$$
E_1 \leq C\left(1-\mathbb{\overline{E}}\left[\frac{\textrm{Card}(\mathsf{S}_{p-1})}{S_l}\right]
+ \mathbb{\overline{E}}\left[\sum_{i\notin\mathsf{S}_{p-1}^l}\frac{\overline{c}^{N_{[p-1,p]}^{i,\star,l}}}{\sum_{j=1}^{S_l}
\underline{c}^{N_{[p-1,p]}^{j,\star,l}}}\right]
\right).
$$
As in the proof of Lemma \ref{lem:3}, using $\underline{c}\overline{c}\geq 1$, along with the conditional Jensen inequality we obtain the upper-bound:
$$
E_1 \leq C\left(1-\mathbb{\overline{E}}\left[\frac{\textrm{Card}(\mathsf{S}_{p-1})}{S_l}\right]
+ 
\mathbb{\overline{E}}\left[\left(\tfrac{1}{S_l}\sum_{i\notin\mathsf{S}_{p-1}^l}\overline{c}^{N_{[p-1,p]}^{i,\star,l}}\right)
\left(\tfrac{1}{S_l}
\sum_{j=1}^{S_l}
\overline{c}^{N_{[p-1,p]}^{j,\star,l}}
\right)\right]
\right).
$$
Again using the independence of the $N_{[p-1,p]}^{j,\star,l}$, one has taking expectations w.r.t.~these random variables that
\begin{equation}\label{eq:lem5}
E_1 \leq C\left(1-\mathbb{\overline{E}}\left[\frac{\textrm{Card}(\mathsf{S}_{p-1})}{S_l}\right]\right).
\end{equation}

For $E_2$, we have the decomposition $E_2\leq E_3+E_4$ where
\begin{eqnarray*}
E_3 & = &  \mathbb{\overline{E}}\left[\sum_{i\in\mathsf{S}_{p-1}^l}\frac{|g(X_p^{i,l},y_p)-g(\overline{X}_p^{i,l-1},y_p)R_{p-1,p}^l(\mathbf{X}_{p-1}^{i,l},\Xi_{p-1,p}^{i,l},\mathbf{X}_p^{i,l})|}{\sum_{j=1}^{S_l}g(X_p^{j,l},y_p)}\right]\\
E_4 & = & \mathbb{\overline{E}}\Bigg[
\sum_{i\in\mathsf{S}_{p-1}^l}\frac{
g(\overline{X}_p^{i,l-1},y_p)R_{p-1,p}^l(\mathbf{X}_{p-1}^{i,l},\Xi_{p-1,p}^{i,l},\mathbf{X}_p^{i,l})
}
{
\{\sum_{j=1}^{S_l}g(X_p^{j,l},y_p)\}
\{\sum_{j=1}^{S_l} g(\overline{X}_p^{j,l-1},y_p)R_{p-1,p}^l(\mathbf{X}_{p-1}^{j,l},\Xi_{p-1,p}^{j,l},\mathbf{X}_p^{j,l})\}
}\times \\ & & 
\sum_{j=1}^{S_l}|g(X_p^{j,l},y_p)-g(\overline{X}_p^{j,l-1},y_p)R_{p-1,p}^l(\mathbf{X}_{p-1}^{j,l},\Xi_{p-1,p}^{j,l},\mathbf{X}_p^{j,l})|
\Bigg].
\end{eqnarray*}
In the case of $E_3$ using (A\ref{ass:2}) we have the upper-bound
$$
E_3 \leq C\mathbb{\overline{E}}\left[\frac{1}{S_l}\sum_{i\in\mathsf{S}_{p-1}^l}|g(X_p^{i,l},y_p)-g(\overline{X}_p^{i,l-1},y_p)R_{p-1,p}^l(\mathbf{X}_{p-1}^{i,l},\Xi_{p-1,p}^{i,l},\mathbf{X}_p^{i,l})|\right]
$$
and so applying Lemma \ref{lem:3} we deduce that
$$
E_3 \leq C\Delta_l.
$$
For $E_4$, using (A\ref{ass:2}), and \eqref{g:bound}  we obtain
$$
E_4\leq C\mathbb{\overline{E}}\Bigg[\Bigg(
\sum_{i\in\mathsf{S}_{p-1}^l}\frac{\overline{c}^{N_{[p-1,p]}^{i,\star,l}}}{\sum_{j=1}^{S_l}
\underline{c}^{N_{[p-1,p]}^{j,\star,l}}}
\Bigg)
\Bigg(
\frac{1}{S_l}\sum_{j=1}^{S_l}|g(X_p^{j,l},y_p)-g(\overline{X}_p^{j,l-1},y_p)R_{p-1,p}^l(\mathbf{X}_{p-1}^{j,l},\Xi_{p-1,p}^{j,l},\mathbf{X}_p^{j,l})|
\Bigg)
\Bigg].
$$
Then splitting the summation over $j$ in the numerator between $\mathsf{S}_{p-1}^l$ and $(\mathsf{S}_{p-1}^l)^c$ and one application of the Cauchy-Schwarz inequality, one can deduce the upper-bound
$E_4\leq E_5+E_6$ where
\begin{eqnarray*}
E_5 & = & \mathbb{\overline{E}}\Bigg[\Bigg(
\sum_{i\in\mathsf{S}_{p-1}^l}\frac{\overline{c}^{N_{[p-1,p]}^{i,\star,l}}}{\sum_{j=1}^{S_l}
\underline{c}^{N_{[p-1,p]}^{j,\star,l}}}
\Bigg)^2\Bigg]^{1/2}\times \\ & & \mathbb{\overline{E}}\left[
\Bigg(
\frac{1}{S_l}\sum_{j\in\mathsf{S}_{p-1}^l}|g(X_p^{j,l},y_p)-g(\overline{X}_p^{j,l-1},y_p)R_{p-1,p}^l(\mathbf{X}_{p-1}^{j,l},\Xi_{p-1,p}^{j,l},\mathbf{X}_p^{j,l})|
\Bigg)^2
\right]^{1/2}\\
E_6 & = & \mathbb{\overline{E}}\Bigg[\Bigg(
\sum_{i\in\mathsf{S}_{p-1}^l}\frac{\overline{c}^{N_{[p-1,p]}^{i,\star,l}}}{\sum_{j=1}^{S_l}
\underline{c}^{N_{[p-1,p]}^{j,\star,l}}}
\Bigg)
\Bigg(
\frac{1}{S_l}\sum_{j\notin\mathsf{S}_{p-1}^l}|g(X_p^{j,l},y_p)-g(\overline{X}_p^{j,l-1},y_p)R_{p-1,p}^l(\mathbf{X}_{p-1}^{j,l},\Xi_{p-1,p}^{j,l},\mathbf{X}_p^{j,l})|
\Bigg)
\Bigg].
\end{eqnarray*}
For $E_5$ the left expectation can be controlled by using similar arguments as adopted for $E_1$ and the right expectation by an application of the Jensen inequality and Lemma \ref{lem:3} to yield:
$$
E_5 \leq C\Delta_l.
$$
For $E_6$, by (A\ref{ass:2}), and \eqref{g:bound} we have the almost sure upper-bound
$$
|g(X_p^{j,l},y_p)-g(\overline{X}_p^{j,l-1},y_p)R_{p-1,p}^l(\mathbf{X}_{p-1}^{j,l},\Xi_{p-1,p}^{j,l},\mathbf{X}_p^{j,l})| \leq
C\left(1+\overline{c}^{N_{[p-1,p]}^{j,\star,l}}\right).
$$
Therefore we have that
$$
E_6 \leq
 C\mathbb{\overline{E}}\Bigg[\Bigg(
\sum_{i\in\mathsf{S}_{p-1}^l}\frac{\overline{c}^{N_{[p-1,p]}^{i,\star,l}}}{\sum_{j=1}^{S_l}
\underline{c}^{N_{[p-1,p]}^{j,\star,l}}}
\Bigg)
\Bigg(1-\frac{\textrm{Card}(\mathsf{S}_{p-1}^l)}{S_l} +
\frac{1}{S_l}\sum_{j\notin\mathsf{S}_{p-1}^l}\overline{c}^{N_{[p-1,p]}^{j,\star,l}}
\Bigg)
\Bigg].
$$
Using similar arguments as used to control $E_1$ one can prove that
$$
E_6 \leq C\left(1-\mathbb{\overline{E}}\left[\frac{\textrm{Card}(\mathsf{S}_{p-1})}{S_l}\right]\right).
$$
Summarizing the above upper-bounds on $E_3,\dots,E_6$, we have shown that
\begin{equation}\label{eq:lem6}
E_2 \leq C\left(1-\mathbb{\overline{E}}\left[\frac{\textrm{Card}(\mathsf{S}_{p-1})}{S_l}\right]+\Delta_l\right).
\end{equation}

Now, noting the relations \eqref{eq:lem1}-\eqref{eq:lem4} and combining these with the bounds 
\eqref{eq:lem5} and \eqref{eq:lem6}, we have shown that
$$
1- \mathbb{\overline{E}}\left[\frac{\textrm{Card}(\mathsf{S}_{p})}{S_l}\right]
\leq C\left(1-\mathbb{\overline{E}}\left[\frac{\textrm{Card}(\mathsf{S}_{p-1})}{S_l}\right]+\Delta_l\right).
$$
Therefore, the proof is concluded by induction.
\end{proof}

\begin{lem}\label{lem:5}
Assume (A\ref{ass:1}-\ref{ass:3}).  Then for any $(\varphi,p,\kappa)\in\mathcal{B}_b(\mathsf{X})\cap\textrm{\emph{Lip}}(\mathsf{X})\times\{1,\dots,T\}\times(0,\infty)$ there exists a $C<+\infty$ such that for any $(l,S_l)\in\{2,3,\dots\}\times\mathbb{N}$:
$$
\mathbb{\overline{E}}\left[|\varphi(X_p^{1,l})-\varphi(\overline{X}_{p}^{1,l-1})
R_{p-1,p}^l(\mathbf{X}_{p-1}^{1,l},\Xi_{p-1,p}^{1,l},\mathbf{X}_{p}^{1,l})|^{\kappa}\right]
\leq C\Delta_l^{\kappa\wedge 1}.
$$
\end{lem}

\begin{proof}
This follows in a similar manner to the proof of \cite[Theorem D.5.]{mlpf}, except we must use Lemmata \ref{lem:3} and \ref{lem:4} of this paper and the arguments employed in the previous proofs (e.g.~Lemma \ref{lem:4}) to control the almost sure upper-bound on $R_{p-1,p}^l$. As these arguments are repetitive, we omit them for brevity.
\end{proof}

%

\subsection{Convergence Proofs}\label{app:conv}

In order to present our subsequent results, we will introduce several notations. 
Writing $g_k(x)=g(x,y_k)$,
we define the sequence of probability measures $\eta_1(dx) = M_1(x_0,dx)$ and for $p\in\{2,\dots,T\}$
$$
\eta_p(dx) = \frac{\int_{\mathsf{X}}\eta_{p-1}(d\overline{x})g_{p-1}(\overline{x})M_p(\overline{x},dx)}{\int_{\mathsf{X}}\eta_{p-1}(d\overline{x})g_{p-1}(\overline{x})}.
$$
The approximate measures $\eta_p^l$ are defined in the same way, except replacing $M_1,\dots,M_T$
with $M_1^l,\dots,M_T^l$. We use the notations 
$\eta_p(\varphi)=\int_{\mathsf{X}}\varphi(x)\eta_p(dx)$, 
$\eta_p^l(\varphi)=\int_{\mathsf{X}}\varphi(x)\eta_p^l(dx)$, 
with $\varphi\in\mathcal{B}_b(\mathsf{X})$. 

For $(p,l)\in\{1,\dots,T\}\times\mathbb{N}$, we introduce the space of $\Xi_{p-1,p}^l$:
$$
\mathsf{T}_p := \bigcup_{k\geq 0}\left(\{k\}\times[p-1,p]^{k}\times\Big\{\bigcup_{q=0}^k\Big[\{q\}\times\{1,\dots,k\}^q\times\mathsf{X}^{2q}\Big]\Big\}\right).
$$
Denote the corresponding $\sigma-$field that is generated by $\mathsf{T}_p$ as $\mathscr{T}_p$.
For $(p,l)\in\{1,\dots,T\}\times\{2,3,\dots\}$, we define $\check{M}_p^l:\mathsf{X}^2\times\mathscr{X}\vee\mathscr{T}_p\vee\mathscr{X}\rightarrow[0,1]$
as the kernel that generates $\Psi_p^l=(\Xi_{p-1,p}^l,\mathbf{X}_{p}^{l})$. Define $\overline{M}_p^l:\mathsf{X}^4\times\mathscr{X}\vee\mathscr{T}_p\vee\mathscr{X}\rightarrow[0,1]$ as:
$$
\overline{M}_p^l\left(\Big((x_{p-1}^l,\overline{x}_{p-1}^{l-1}),(\check{x}_{p-1}^l,\check{\overline{x}}_{p-1}^{l-1})\Big),d\psi_p^l\right) = \check{M}_p^l\left((x_{p-1}^l,\check{\overline{x}}_{p-1}^{l-1}),d\psi_p^l\right)
$$
where $\Big((x_{p-1}^l,\overline{x}_{p-1}^{l-1}),(\check{x}_{p-1}^l,\check{\overline{w}}_{p-1}^{l-1})\Big)\in\mathsf{X}^2\times\mathsf{X}^2$.
Then for $\mu$ a probability measure on $\mathsf{E}_{p-1}:=\mathsf{X}^2\times\mathsf{T}_{p-1}\times\mathsf{X}^2$ 
(which generates the $\sigma-$field $\mathscr{E}_{p-1}$)
and 
$(p,l)\in\{2,\dots,T\}\times\{2,3,\dots\}$ we write the probability measure on $(\mathsf{E}_p,\mathscr{E}_p)$
$$
\check{\Theta}_p^{l}(\mu)(d(\tilde{\mathbf{x}}_{p-1}^l,\psi_p^l)) :=  
$$
$$
\int_{\mathsf{E}_{p-1}} 
\{G_{p-1,\mu}^l(x_{p-1}^l)\wedge \overline{G}_{p-1,\mu}^{l-1}(\mathbf{x}_{p-2}^l,\psi_{p-1}^l)\}
\check{M}_p^l\left((x_{p-1}^l,\overline{x}_{p-1}^{l-1}),d\psi_p^l\right)
\delta_{\{\mathbf{x}_{p-1}^l\}}(d\tilde{\mathbf{x}}_{p-1}^l)
\mu(d(\mathbf{x}_{p-2}^l,\psi_{p-1}^l)) + 
$$
$$
\Bigg(1-
\int_{\mathsf{E}_{p-1}} 
\{G_{p-1,\mu}^l(x_{p-1}^l)\wedge \overline{G}_{p-1,\mu}^{l-1}(\mathbf{x}_{p-2}^l,\psi_{p-1}^l)\}
\mu(d(\mathbf{x}_{p-2}^l,\psi_{p-1}^l))\Bigg)
\int_{\mathsf{E}_{p-1}\times \mathsf{E}_{p-1}} \check{G}_{p-1,\mu}^l(\mathbf{x}_{p-2}^l,\psi_{p-1}^l)
\check{\overline{G}}_{p-1,\mu}^{l-1}(\check{\mathbf{x}}_{p-2}^l,\check{\psi}_{p-1}^{l}) 
$$
$$
\overline{M}_p^l\left(\Big((x_{p-1}^l,\overline{x}_{p-1}^{l-1}),(\check{x}_{p-1}^l,\check{\overline{x}}_{p-1}^{l-1})\Big),d\psi_p^l\right)\delta_{\{(x_{p-1}^l,\check{\overline{x}}_{p-1}^l)\}}(d\tilde{\mathbf{x}}_{p-1}^{l})
\mu(d(\mathbf{x}_{p-2}^l,\psi_{p-1}^l))\mu(d(\check{\mathbf{x}}_{p-2}^l,\check{\psi}_{p-1}^{l})))
$$
where
\begin{eqnarray*}
\check{G}_{p-1,\mu}^l(\mathbf{x}_{p-2}^l,\psi_{p-1}^l) & = & \frac{G_{p-1,\mu}^l(x_{p-1}^l)-G_{p-1,\mu}^l(x_{p-1}^l)\wedge \overline{G}_{p-1,\mu}^{l-1}(\mathbf{x}_{p-2}^l,\psi_{p-1}^l)}{\int_{\mathsf{E}_{p-1}}\{G_{p-1,\mu}^l(x_{p-1}^l)-G_{p-1,\mu}^l(x_{p-1}^l)\wedge \overline{G}_{p-1,\mu}^{l-1}(\mathbf{x}_{p-2}^l,\psi_{p-1}^l)\}\mu(d(\mathbf{x}_{p-2}^l,\psi_{p-1}^l))} \\
\check{\overline{G}}_{p-1,\mu}^{l-1}(\mathbf{x}_{p-2}^l,\psi_{p-1}^l) & = & 
\frac{\overline{G}_{p-1,\mu}^{l-1}(\mathbf{x}_{p-2}^l,\psi_{p-1}^l)-G_{p-1,\mu}^l(x_{p-1}^l)\wedge \overline{G}_{p-1,\mu}^{l-1}(\mathbf{x}_{p-2}^l,\psi_{p-1}^l)}{\int_{\mathsf{E}_{p-1}}\{\overline{G}_{p-1,\mu}^{l-1}(\mathbf{x}_{p-2}^l,\psi_{p-1}^l)-G_{p-1,\mu}^l(x_{p-1}^l)\wedge \overline{G}_{p-1,\mu}^{l-1}(\mathbf{x}_{p-2}^l,\psi_{p-1}^l)\}\mu(d(\mathbf{x}_{p-2}^l,\psi_{p-1}^l))}\\
G_{p-1,\mu}^l(x_{p-1}^l) & = & \frac{g_{p-1}(x_{p-1}^l)}{\int_{\mathsf{E}_{p-1}}g_{p-1}(x_{p-1}^l)\mu(d(\mathbf{x}_{p-2}^l,\psi_{p-1}^l))} \\
\overline{G}_{p-1,\mu}^{l-1}(\mathbf{x}_{p-2}^l,\psi_{p-1}^l) & = & \frac{g_{p-1}(\overline{x}_{p-1}^{l-1})R_{p-2,p-1}^l(\mathbf{x}_{p-2}^l,\psi_{p-1}^l)}{\int_{\mathsf{E}_{p-1}}g_{p-1}(\overline{x}_{p-1}^{l-1})R_{p-2,p-1}^l(\mathbf{x}_{p-2}^l,\psi_{p-1}^l)\mu(d(\mathbf{x}_{p-2}^l,\psi_{p-1}^l))}
\end{eqnarray*}
and $\delta_{\{\mathbf{x}\}}(d\tilde{\mathbf{x}})$ is the Dirac measure.
The probability measure $\check{\Theta}_p^{l}(\mu)$ represents the resampling (maximal coupling) and sampling operation in Algorithm \ref{alg:mlpf} and was first derived in \cite{coup_clt}. We will explain, below, why it has been introduced.

We define the following sequence of probability 
measures on $\mathsf{E}_1,\dots,\mathsf{E}_T$: 
$$
\check{\eta}_1^l(d(\tilde{\mathbf{x}}_0^l,\psi_1^l)) = \check{M}_p^l((x_0,x_0),d\psi_1^l)\delta_{\{x_0,x_0\}}(d\tilde{\mathbf{x}}_0^l)
$$ 
and then for $p\in\{1,\dots,T-1\}$
$$
\check{\eta}_{p+1}^l\Big((d(\tilde{\mathbf{x}}_{p}^l,\psi_{p+1}^l))\Big)=\check{\Theta}^l_p(\check{\eta}_{p}^l)\Big((d(\tilde{\mathbf{x}}_{p}^l,\psi_{p+1}^l))\Big).
$$ 
It is easily checked that for any $p\in\{1,\dots,T\}$ the marginal of $\check{\eta}_{p}^l$ in the $x_p^l$ co-ordinate is $\eta_p^l$. In addition, one can show that for $(\varphi,p)\in\mathcal{B}_b(\mathsf{X})\times\{1,\dots,T\}$
$$
\pi^{l-1}_p(\varphi) = \frac{\int_{\mathsf{E}_p}\varphi(\overline{x}_p^{l-1})R_{p-1,p}^l(\mathbf{x}_{p-1}^l,\psi_p^l)g_p(\overline{x}_{p}^{l-1})\check{\eta}^l_p(d(\mathbf{x}_{p-1}^l,\psi_p^l))}{\int_{\mathsf{E}_p}R_{p-1,p}^l(\mathbf{x}_{p-1}^l,\psi_p^l)g_p(\overline{x}_{p}^{l-1})\check{\eta}^l_p(d(\mathbf{x}_{p-1}^l,\psi_p^l))}.
$$
Now consider Algorithm \ref{alg:mlpf} run at a level $l\in\{2,3,\dots\}$ and for $(\varphi,p)\in\mathcal{B}_b(\mathsf{X})\times\{1,\dots,T\}$ and define
\begin{eqnarray*}
\pi_p^{S_l,l}(\varphi) & := & \frac{\tfrac{1}{S_l}\sum_{i=1}^{S_l}\varphi(X_p^{i,l})g_p(X_p^{i,l})}{\tfrac{1}{S_l}\sum_{i=1}^{S_l}
g_p(X_p^{i,l})} \\
\overline{\pi}_p^{S_l,l-1}(\varphi) & := & \frac{\tfrac{1}{S_l}\sum_{i=1}^{S_l}\varphi(\overline{X}_p^{i,l-1})
R_{p-1,p}^l(\mathbf{X}_{p-1}^{i,l},\Psi_p^{i,l})
g_p(\overline{X}_p^{i,l-1})}{\tfrac{1}{S_l}\sum_{i=1}^{S_l}
R_{p-1,p}^l(\mathbf{X}_{p-1}^{i,l},\Psi_p^{i,l})
g_p(\overline{X}_p^{i,l-1})}.
\end{eqnarray*}
where we have used the notation $\Psi_p^{i,l}=(\Xi_{p-1,p}^{i,l},\mathbf{X}_p^{i,l})$.
Then the estimate of the level difference $\pi_p^l(\varphi)-\pi_p^{l-1}(\varphi)$ is then $\pi_p^{S_l,l}(\varphi)-\overline{\pi}_p^{S_l,l-1}(\varphi)$. 

The main barrier to a complete analysis of the estimate $\pi_p^{S_l,l}(\varphi)-\overline{\pi}_p^{S_l,l-1}(\varphi)$ is the complicated structure of the operator $\check{\Theta}_p^l$. Essentially, for $\kappa>0$,  to provide a bound (i.e.~$\mathcal{O}(S_l^{-\kappa/2})$) on quantities of the type
\begin{equation}\label{eq:exp_exp}
\mathbb{\overline{E}}\left[\left|\frac{1}{S_l}\sum_{i=1}^{S_l}\varphi(\overline{X}_p^{i,l-1})
R_{p-1,p}^l(\mathbf{X}_{p-1}^{i,l},\Psi_p^{i,l})
g_p(\overline{X}_p^{i,l-1})-\int_{\mathsf{E}_p}\varphi(\overline{x}_p^{l-1})R_{p-1,p}^l(\mathbf{x}_{p-1}^l,\psi_p^l)g_p(\overline{x}_{p}^{l-1})\check{\eta}^l_p(d(\mathbf{x}_{p-1}^l,\psi_p^l))\right|^{\kappa}\right]
\end{equation}
one typically works via a proof by induction on $p$. The case $p=1$ is typically straightforward using the Marcinkiewicz-Zygmund inequality, at least when $\kappa\geq 2$, but the case $\kappa\in(0,2)$ can be recovered via the bound when $\kappa=2$ and Jensen. The induction is particularly troublesome as one typically will add and subtract the conditional expectation:
$$
\mathbb{\overline{E}}\left[\frac{1}{S_l}\sum_{i=1}^{S_l}\varphi(\overline{X}_p^{i,l-1})
R_{p-1,p}^l(\mathbf{X}_{p-1}^{i,l},\Psi_p^{i,l})
g_p(\overline{X}_p^{i,l-1})\Big|\mathscr{F}_{p-1}^{S^l,l}\right]
$$
in the argument $|\cdot|^{\kappa}$ in the expectation operator and where $\mathscr{F}_{p-1}^{S^l,l}$ is the natural filtration generated by the particle system at time $p-1$, (after sampling, we use $\mathscr{F}_{0}^{S^l,l}$ to denote the trivial $\sigma-$ field). This latter conditional expectation can be written exactly in terms of the operator $\check{\Theta}_p^l$ where the input measure is a particular empirical measure associated to the particle system; we omit the exact details for brevity. The problem therein lies with the fact that  $\check{\Theta}_p^l$ will depend on a product measure of the afore-mentioned empirical measure and as such the induction hypothesis (that the expectation in \eqref{eq:exp_exp} is bounded by a term $\mathcal{O}(S_l^{-\kappa/2})$) is invalidated. This was realized in the work of \cite{coup_clt} where a limit theorem (convergence in probability) is given (see \cite[Theorem 3.1]{coup_clt}). To overcome the issue of the product measure, a density argument based upon the Stone-Wierstrass theorem is used. That argument can be extended to the case of the algorithm in this article, \emph{only} if $R_{p-1,p}^l$ is continuous and bounded; as established in \eqref{r:bound} the latter condition does not hold. As a result as it is unclear how to extend the proofs in \cite{coup_clt}, we decide to make an additional assumption, which are results that need to be proved, in order to provide formal bounds. This assumption is as follows and we denote convergence in probability as $S_l\rightarrow+\infty$ as $\rightarrow_{\mathbb{P}}$.

\begin{hypA}\label{ass:5}
For any $(\varphi,p,l)\in\mathcal{B}_b(\mathsf{X})\cap\textrm{Lip}(\mathsf{X})\times\{1,\dots,T\}\times\{2,3\dots\}$
$$
\mathbb{\overline{E}}\left[\frac{1}{S_l}\sum_{i=1}^{S_l}\varphi(\overline{X}_p^{i,l-1})
R_{p-1,p}^l(\mathbf{X}_{p-1}^{i,l},\Psi_p^{i,l})
g_p(\overline{X}_p^{i,l-1})\Big|\mathscr{F}_{p-1}^{S^l,l}\right] \rightarrow_{\mathbb{P}}
$$
$$
\int_{\mathsf{E}_p}\varphi(\overline{x}_p^{l-1})R_{p-1,p}^l(\mathbf{x}_{p-1}^l,\psi_p^l)g_p(\overline{x}_{p}^{l-1})\check{\eta}^l_p(d(\mathbf{x}_{p-1}^l,\psi_p^l)).
$$
\end{hypA}
The reason that we express the assumption as a convergence in probability, instead of a bound on a moment (which is what is needed) is that we expect that, following the arguments in \cite[Theorem 3.1]{coup_clt} that (A\ref{ass:5}) can be proved. However, we would expect that to prove a bound on a moment would be extremely arduous and thus the form of the assumption given.

\begin{lem}\label{lem:6}
Assume (A\ref{ass:1}-\ref{ass:5}).  Then for any $(\varphi,p,\kappa)\in\mathcal{B}_b(\mathsf{X})\cap\textrm{\emph{Lip}}(\mathsf{X})\times\{1,\dots,T\}\times(0,\infty)$ there exists a $C<+\infty$ such that for any 
$\epsilon>0$ there exists a $S_l^{\epsilon}\in\mathbb{N}$ so that for
$l\in\{2,3,\dots\}$:
$$
\mathbb{\overline{E}}\left[|
\pi_p^{S_l^{\epsilon},l}(\varphi)-\overline{\pi}_p^{S_l^{\epsilon},l-1}(\varphi) - \{\pi_p^l(\varphi)-\pi_p^{l-1}(\varphi)\}
|^{\kappa}\right]^{1/\kappa}
\leq C\left(\frac{\Delta_l^{\kappa^{-1}\wedge 1}}{\sqrt{S_l^{\epsilon}}}+\epsilon\right).
$$
\end{lem}

\begin{proof}
Throughout, we assume $\kappa\geq 2$; the case $\kappa\in(0,2)$ can be obtained by using the bound when $\kappa=2$ and applying Jensen's inequality.
We begin by noting the following decomposition, which can be established via \cite[Lemma C.5.]{mlpf}, for any $S_l\in\mathbb{N}$:
\begin{equation}\label{eq:decomp_lp}
\pi_p^{S_l,l}(\varphi)-\overline{\pi}_p^{S_l,l-1}(\varphi) - \{\pi_p^l(\varphi)-\pi_p^{l-1}(\varphi)\} = \sum_{j=1}^6 E_j
\end{equation}
where
\begin{eqnarray*}
E_1 & = & \frac{1}{\tfrac{1}{S_l}\sum_{i=1}^{S_l} g_p(X_p^{i,l})}\Bigg(
\frac{1}{S_l}\sum_{i=1}^{S_l}\varphi(X_p^{i,l})g_p(X_p^{i,l}) - 
\frac{1}{S_l}\sum_{i=1}^{S_l}\varphi(\overline{X}_p^{i,l-1})
R_{p-1,p}^l(\mathbf{X}_{p-1}^{i,l},\Psi_p^{i,l})
g_p(\overline{X}_p^{i,l-1}) - \\ & &  \Big\{
\int_{\mathsf{X}}\varphi(x)g_p(x)\eta_p^l(dx) -
\int_{\mathsf{E}_p}\varphi(\overline{x}_p^{l-1})R_{p-1,p}^l(\mathbf{x}_{p-1}^l,\psi_p^l)g_p(\overline{x}_{p}^{l-1})\check{\eta}^l_p(d(\mathbf{x}_{p-1}^l,\psi_p^l))
\Big\}
\Bigg)\\
E_2 & = & -\frac{\tfrac{1}{S_l}\sum_{i=1}^{S_l}\varphi(\overline{X}_p^{i,l-1})
R_{p-1,p}^l(\mathbf{X}_{p-1}^{i,l},\Psi_p^{i,l})
g_p(\overline{X}_p^{i,l-1})}{\{\tfrac{1}{S_l}\sum_{i=1}^{S_l} g_p(X_p^{i,l})\}\{
\tfrac{1}{S_l}\sum_{i=1}^{S_l}
R_{p-1,p}^l(\mathbf{X}_{p-1}^{i,l},\Psi_p^{i,l})
g_p(\overline{X}_p^{i,l-1})
\}}
\Bigg(\frac{1}{S_l}\sum_{i=1}^{S_l}g_p(X_p^{i,l}) - \\ & &
\frac{1}{S_l}\sum_{i=1}^{S_l}
R_{p-1,p}^l(\mathbf{X}_{p-1}^{i,l},\Psi_p^{i,l})
g_p(\overline{X}_p^{i,l-1}) -  
 \Big\{
\int_{\mathsf{X}}g_p(x)\eta_p^l(dx) -
\int_{\mathsf{E}_p}R_{p-1,p}^l(\mathbf{x}_{p-1}^l,\psi_p^l)g_p(\overline{x}_{p}^{l-1})\times \\ & & \check{\eta}^l_p(d(\mathbf{x}_{p-1}^l,\psi_p^l))
\Big\}
\Bigg)
\\
E_3 & = & \frac{1}{\{\tfrac{1}{S_l}\sum_{i=1}^{S_l} g_p(X_p^{i,l})\}\int_{\mathsf{X}}g_p(x)\eta_p^l(dx)}
\Big(\int_{\mathsf{X}}g_p(x)\eta_p^l(dx)-\frac{1}{S_l}\sum_{i=1}^{S_l} g_p(X_p^{i,l})\Big)
\Big(
\int_{\mathsf{X}}\varphi(x)g_p(x)\eta_p^l(dx) - \\ & & \int_{\mathsf{E}_p}\varphi(\overline{x}_p^{l-1})R_{p-1,p}^l(\mathbf{x}_{p-1}^l,\psi_p^l)g_p(\overline{x}_{p}^{l-1})\check{\eta}^l_p(d(\mathbf{x}_{p-1}^l,\psi_p^l))
\Big)\\
E_4 &= & -\frac{1}{\{\tfrac{1}{S_l}\sum_{i=1}^{S_l} g_p(X_p^{i,l})\}\{
\tfrac{1}{S_l}\sum_{i=1}^{S_l}
R_{p-1,p}^l(\mathbf{X}_{p-1}^{i,l},\Psi_p^{i,l})
g_p(\overline{X}_p^{i,l-1})
\}}\Big(\frac{1}{S_l}\sum_{i=1}^{S_l}\varphi(\overline{X}_p^{i,l-1})
R_{p-1,p}^l(\mathbf{X}_{p-1}^{i,l},\Psi_p^{i,l})
g_p(\overline{X}_p^{i,l-1})- \\ & &
\int_{\mathsf{E}_p}\varphi(\overline{x}_p^{l-1})R_{p-1,p}^l(\mathbf{x}_{p-1}^l,\psi_p^l)g_p(\overline{x}_{p}^{l-1})\check{\eta}^l_p(d(\mathbf{x}_{p-1}^l,\psi_p^l))
\Big)\Big(\int_{\mathsf{X}}g_p(x)\eta_p^l(dx)
-\\ & &\int_{\mathsf{E}_p}R_{p-1,p}^l(\mathbf{x}_{p-1}^l,\psi_p^l)g_p(\overline{x}_{p}^{l-1})\check{\eta}^l_p(d(\mathbf{x}_{p-1}^l,\psi_p^l))
\Big)\\
E_5 & = & \frac{\int_{\mathsf{E}_p}\varphi(\overline{x}_p^{l-1})R_{p-1,p}^l(\mathbf{x}_{p-1}^l,\psi_p^l)g_p(\overline{x}_{p}^{l-1})\check{\eta}^l_p(d(\mathbf{x}_{p-1}^l,\psi_p^l))
}{
\{\tfrac{1}{S_l}\sum_{i=1}^{S_l}
R_{p-1,p}^l(\mathbf{X}_{p-1}^{i,l},\Psi_p^{i,l})
g_p(\overline{X}_p^{i,l-1})\}
\int_{\mathsf{X}}g_p(x)\eta_p^l(dx)\int_{\mathsf{E}_p}R_{p-1,p}^l(\mathbf{x}_{p-1}^l,\psi_p^l)g_p(\overline{x}_{p}^{l-1})\check{\eta}^l_p(d(\mathbf{x}_{p-1}^l,\psi_p^l))
}\times \\ & &\Big(
\frac{1}{S_l}\sum_{i=1}^{S_l}
R_{p-1,p}^l(\mathbf{X}_{p-1}^{i,l},\Psi_p^{i,l})
g_p(\overline{X}_p^{i,l-1}) -
\int_{\mathsf{X}}g_p(x)\eta_p^l(dx)\int_{\mathsf{E}_p}R_{p-1,p}^l(\mathbf{x}_{p-1}^l,\psi_p^l)g_p(\overline{x}_{p}^{l-1})\check{\eta}^l_p(d(\mathbf{x}_{p-1}^l,\psi_p^l))
\Big)\times \\ & &
\Big(\int_{\mathsf{X}}g_p(x)\eta_p^l(dx)
-\int_{\mathsf{E}_p}R_{p-1,p}^l(\mathbf{x}_{p-1}^l,\psi_p^l)g_p(\overline{x}_{p}^{l-1})\check{\eta}^l_p(d(\mathbf{x}_{p-1}^l,\psi_p^l))
\Big)\\
E_6 & = & -\frac{\int_{\mathsf{E}_p}\varphi(\overline{x}_p^{l-1})R_{p-1,p}^l(\mathbf{x}_{p-1}^l,\psi_p^l)g_p(\overline{x}_{p}^{l-1})\check{\eta}^l_p(d(\mathbf{x}_{p-1}^l,\psi_p^l))
}{
\{\tfrac{1}{S_l}\sum_{i=1}^{S_l} g_p(X_p^{i,l})\}
\{\tfrac{1}{S_l}\sum_{i=1}^{S_l}
R_{p-1,p}^l(\mathbf{X}_{p-1}^{i,l},\Psi_p^{i,l})
g_p(\overline{X}_p^{i,l-1})\}\int_{\mathsf{X}}g_p(x)\eta_p^l(dx)
}\Big(\int_{\mathsf{X}}g_p(x)\eta_p^l(dx)-\\ & &\frac{1}{S_l}\sum_{i=1}^{S_l} g_p(X_p^{i,l})\Big)
\Big(\int_{\mathsf{X}}g_p(x)\eta_p^l(dx)
-\int_{\mathsf{E}_p}R_{p-1,p}^l(\mathbf{x}_{p-1}^l,\psi_p^l)g_p(\overline{x}_{p}^{l-1})\check{\eta}^l_p(d(\mathbf{x}_{p-1}^l,\psi_p^l))
\Big).
\end{eqnarray*}
The proof, via Minkowski, can be bounding the $\mathbb{L}_{\kappa}-$norms of each of the terms $E_1,\dots,E_6$. As $E_1$ and $E_2$ are similar we give the proof for $E_1$ only. Similarly $E_3$ (resp.~$E_4$) and $E_6$  (resp.~$E_5$) are almost the same proofs, so we consider $E_3$ (resp.~$E_4$) only.

For the term $E_1$, using (A\ref{ass:2}) and adding and subtracting the conditional expectation:
$$
\overline{\mathbb{E}}\left[\frac{1}{S_l}\sum_{i=1}^{S_l}\varphi(X_p^{i,l})g_p(X_p^{i,l}) - 
\frac{1}{S_l}\sum_{i=1}^{S_l}\varphi(\overline{X}_p^{i,l-1})
R_{p-1,p}^l(\mathbf{X}_{p-1}^{i,l},\Psi_p^{i,l})
g_p(\overline{X}_p^{i,l-1})\Big|\mathscr{F}_{p-1}^{S_l,l}\right]
$$
and applying Minkowski, one has the upper-bound:
$$
\overline{\mathbb{E}}[|E_1|^{\kappa}]^{1/\kappa} \leq C(E_7+E_8)
$$
where
\begin{eqnarray*}
E_7 & = &
\overline{\mathbb{E}}\Bigg[\Bigg|
\frac{1}{S_l}\sum_{i=1}^{S_l}\varphi(X_p^{i,l})g_p(X_p^{i,l}) - 
\frac{1}{S_l}\sum_{i=1}^{S_l}\varphi(\overline{X}_p^{i,l-1})
R_{p-1,p}^l(\mathbf{X}_{p-1}^{i,l},\Psi_p^{i,l})
g_p(\overline{X}_p^{i,l-1}) - \\ & &
\overline{\mathbb{E}}\left[\frac{1}{S_l}\sum_{i=1}^{S_l}\varphi(X_p^{i,l})g_p(X_p^{i,l}) - 
\frac{1}{S_l}\sum_{i=1}^{S_l}\varphi(\overline{X}_p^{i,l-1})
R_{p-1,p}^l(\mathbf{X}_{p-1}^{i,l},\Psi_p^{i,l})
g_p(\overline{X}_p^{i,l-1})\Big|\mathscr{F}_{p-1}^{S_l,l}\right]
\Bigg|^{\kappa}\Bigg]^{1/\kappa} \\
E_8 & = & \overline{\mathbb{E}}\Bigg[\Bigg|
\overline{\mathbb{E}}\left[\frac{1}{S_l}\sum_{i=1}^{S_l}\varphi(X_p^{i,l})g_p(X_p^{i,l}) - 
\frac{1}{S_l}\sum_{i=1}^{S_l}\varphi(\overline{X}_p^{i,l-1})
R_{p-1,p}^l(\mathbf{X}_{p-1}^{i,l},\Psi_p^{i,l})
g_p(\overline{X}_p^{i,l-1})\Big|\mathscr{F}_{p-1}^{S_l,l}\right] -\\
& &
\Big\{
\int_{\mathsf{X}}\varphi(x)g_p(x)\eta_p^l(dx) -
\int_{\mathsf{E}_p}\varphi(\overline{x}_p^{l-1})R_{p-1,p}^l(\mathbf{x}_{p-1}^l,\psi_p^l)g_p(\overline{x}_{p}^{l-1})\check{\eta}^l_p(d(\mathbf{x}_{p-1}^l,\psi_p^l))
\Big\}
\Bigg|^{\kappa}\Bigg]^{1/\kappa}.
\end{eqnarray*}
For the term $E_7$ one can use the conditional Marcinkiewicz-Zygmund inequality and Lemma \ref{lem:5}
to obtain that for any $S_l\in\mathbb{N}$
$$
E_7 \leq \frac{C\Delta_l^{\kappa^{-1}\wedge 1}}{\sqrt{S_l}}.
$$
For $E_8$, as one can prove easily (e.g.~\cite[Proposition C.6]{mlpf}) that 
$$
\overline{\mathbb{E}}\left[\frac{1}{S_l}\sum_{i=1}^{S_l}\varphi(X_p^{i,l})g_p(X_p^{i,l})\Big|\mathscr{F}_{p-1}^{S_l,l}\right]
\rightarrow_{\mathbb{P}} \int_{\mathsf{X}}\varphi(x)g_p(x)\eta_p^l(dx)
$$
and that by the arguments in Appendix \ref{app:rate} (see e.g.~\eqref{r:bound}) that for any $\overline{\kappa}>0$
$$
\sup_{S_l\in\mathbb{N}}\overline{\mathbb{E}}\Bigg[\Bigg|
\overline{\mathbb{E}}\left[\frac{1}{S_l}\sum_{i=1}^{S_l}\varphi(X_p^{i,l})g_p(X_p^{i,l}) - 
\frac{1}{S_l}\sum_{i=1}^{S_l}\varphi(\overline{X}_p^{i,l-1})
R_{p-1,p}^l(\mathbf{X}_{p-1}^{i,l},\Psi_p^{i,l})
g_p(\overline{X}_p^{i,l-1})\Big|\mathscr{F}_{p-1}^{S_l,l}\right] -
$$
$$
\Big\{
\int_{\mathsf{X}}\varphi(x)g_p(x)\eta_p^l(dx) -
\int_{\mathsf{E}_p}\varphi(\overline{x}_p^{l-1})R_{p-1,p}^l(\mathbf{x}_{p-1}^l,\psi_p^l)g_p(\overline{x}_{p}^{l-1})\check{\eta}^l_p(d(\mathbf{x}_{p-1}^l,\psi_p^l))
\Big\}
\Bigg|^{\kappa(1+\overline{\kappa})}\Bigg] <+\infty.
$$
Therefore, via (A\ref{ass:5}) 
$$
\lim_{S_l\rightarrow+\infty}E_8 = 0.
$$
As a result, for any $\epsilon>0$ there exists a $S_l^{\epsilon}\in\mathbb{N}$ so that
$$
E_8\leq \epsilon.
$$
This completes the bound for $E_1$.

For $E_3$, as one can prove easily (e.g.~\cite[Proposition C.6]{mlpf}) that
$$
\Bigg|\int_{\mathsf{X}}g_p(x)\eta_p^l(dx)-\frac{1}{S_l}\sum_{i=1}^{S_l} g_p(X_p^{i,l})\Bigg| \rightarrow_{\mathbb{P}} 0
$$
and that for any $\overline{\kappa}>0$
$$
\sup_{S_l\in\mathbb{N}}\overline{\mathbb{E}}[|E_3|^{\kappa(1+\overline{\kappa})}]<+\infty
$$
thus for any $\epsilon>0$ there exists a $S_l^{\epsilon}\in\mathbb{N}$ so that
$$
\overline{\mathbb{E}}[|E_3|^{\kappa}]^{1/\kappa}\leq \epsilon.
$$
For $E_4$ one can use a combination of the approach for $E_1$ and $E_3$; we do not give details for brevity.
This completes the proof.
\end{proof}

\begin{lem}\label{lem:7}
Assume (A\ref{ass:1}-\ref{ass:5}).  Then for any $(\varphi,p,\kappa)\in\mathcal{B}_b(\mathsf{X})\cap\textrm{\emph{Lip}}(\mathsf{X})\times\{1,\dots,T\}\times(0,\infty)$ there exists a $C<+\infty$ such that for any 
$\epsilon>0$ there exists a $S_l^{\epsilon}\in\mathbb{N}$ so that for
$l\in\{2,3,\dots\}$:
$$
\left|\mathbb{\overline{E}}\left[
\pi_p^{S_l^{\epsilon},l}(\varphi)-\overline{\pi}_p^{S_l^{\epsilon},l-1}(\varphi) - \{\pi_p^l(\varphi)-\pi_p^{l-1}(\varphi)\}
\right]\right|
\leq C\epsilon.
$$
\end{lem}

\begin{proof}
One can use the same decomposition \eqref{eq:decomp_lp} and then show, using the approach in the proof of Lemma \ref{lem:6} that the expectation of each of the expressions $E_1,\dots,E_6$ converge to zero. As the arguments are repetitive, they are omitted.
\end{proof}

\subsection{Bias Bounds}\label{app:bias}

\begin{hypA}\label{ass:4}
For any $(\varphi,p)\in\mathcal{B}_b(\mathsf{X})\cap\textrm{Lip}(\mathsf{X})\times\{1,\dots,T\}$
there exists a $C<+\infty$ such that for any $(l,x)\in\mathbb{N}\times\mathsf{X}$:
$$
\left|\int_{\mathsf{X}}\varphi(\overline{x})M_p(x,d\overline{x})-\int_{\mathsf{X}}\varphi(\overline{x})M_p^l(x,d\overline{x})\right| \leq C\Delta_l.
$$
\end{hypA}

On the basis of \cite[Theorem 4.1.]{lemaire} this seems a very reasonable assumption.

\begin{lem}\label{lem:8}
Assume (A\ref{ass:2},\ref{ass:4}).   Then for any $(\varphi,p)\in\mathcal{B}_b(\mathsf{X})\cap\textrm{\emph{Lip}}(\mathsf{X})\times\{1,\dots,T\}$ there exists a $C<+\infty$ such that for any $l\in\mathbb{N}$:
$$
\left|\pi_p(\varphi)-\pi_p^l(\varphi)\right|\leq C\Delta_l.
$$
\end{lem}

\begin{proof}
This follows along the lines of the proof of \cite[Lemma D.2.]{mlpf} and is hence omitted.
\end{proof}

\end{document}